\documentclass[10pt, conference, letterpaper]{IEEEtran}
\IEEEoverridecommandlockouts

\usepackage{cite}
\usepackage{amsmath,amssymb,amsfonts}

\usepackage{graphicx}
\usepackage{textcomp}
\usepackage{xcolor}

\usepackage[ruled,linesnumbered]{algorithm2e}
\usepackage{bbm}
\usepackage{comment}
\usepackage{caption}
\usepackage{subcaption}

\newtheorem{theorem}{Theorem}

\newtheorem{lemma}{Lemma}
\newenvironment{proof}{{\noindent\it Proof.}\,}{\hfill $\square$\par}
\def\BibTeX{{\rm B\kern-.05em{\sc i\kern-.025em b}\kern-.08em
    T\kern-.1667em\lower.7ex\hbox{E}\kern-.125emX}}

\begin{document}

\title{Result and Congestion Aware Optimal Routing and Partial Offloading in Collaborative Edge Computing}

\author{
\IEEEauthorblockN{Jinkun Zhang, Yuezhou Liu, Edmund Yeh}
\IEEEauthorblockA{Department of Electrical and Computer Engineering, Northeastern University,
Boston, US\\
zhang.jinku@northeastern.edu, liu.yuez@northeastern.edu, eyeh@ece.neu.edu}
}%

\maketitle

\vspace{1cm}


\begin{abstract}
Collaborative edge computing (CEC) is an emerging paradigm where heterogeneous edge devices (stakeholders) collaborate to fulfill computation tasks, such as model training or video processing, by sharing communication and computation resources. 
Nevertheless, the optimal data/result routing and computation offloading strategy in CEC with arbitrary topology still remains an open problem. 
In this paper, we formulate a partial-offloading and multi-hop routing model for arbitrarily divisible tasks.
Each node individually decides the computation of the received data and the forwarding of data/result traffic.
In contrast to most existing works, our model applies for tasks with non-negligible result size, and enables separable data sources and result destination.
We propose a network-wide cost minimization problem with congestion-aware cost to jointly optimize routing and computation offloading. This problem covers various performance metrics and constraints, such as average queueing delay with limited processor capacity.
Although the problem is non-convex, we provide non-trivial necessary and sufficient conditions for the global-optimal solution, and devise a fully distributed algorithm that converges to the optimum in polynomial time, 
allows asynchronous individual updating, and is adaptive to changes in network topology or task pattern.
Numerical evaluation shows that our proposed method significantly outperforms other baseline algorithms in multiple network instances, especially in congested scenarios.

\end{abstract}



\section{Introduction} \label{Section:Intro}

In recent years, we are experiencing an explosive increment in the number of mobile and IoT devices. Our daily life is increasingly involved with new mobile applications for work, entertainment, social networking, health care, etc. Many of them are computation-intensive and time-critical, such as AR/VR and autonomous driving. Meanwhile, mobile devices running these applications generate a huge amount of data traffic which is predicted to reach 288EB per month in 2027 \cite{Ericsson2021report}. 
It becomes impractical to direct all computation requests and their data to central cloud due to limited backhaul bandwidth and high latency.
Edge computing is then proposed as a promising solution to provide computation resources and cloud-like services in close proximity to the mobile devices.   

In edge computing, requesters offload their computation to the edge servers. A new concept, extending this idea, is the so-called collaborative edge computing (CEC). Besides point-to-point offloading, CEC let multiple stakeholders (mobile devices, IoT devices, edge servers, or cloud) collaborate with each other by sharing data, communication resources, and computation resources to finish computation tasks \cite{sahni2020multi}. On the one hand, CEC improves the utilization efficiency of resources so that computation-intensive and time-critical services can be better completed at the edge. Mobile devices equipped with computation capabilities can collaborate with each other through D2D communication \cite{sahni2017edge}. Edge servers can also collaborate with each other 
for load balancing or further with central cloud to offload demands that they cannot accommodate \cite{zhu2017socially}. On the other hand, CEC is needed when there is no direct connection between devices and edge servers. Consider unmanned aerial vehicle (UAV) swarms or autonomous cars in rural areas, computation-intensive tasks of UAVs or cars far away from the wireless access point should be collaborated computed or offloaded through multi-hop routing to the edge server with the help of other devices \cite{hong2019multi,sahni2017edge}. 

We wish to study a general framework of CEC which enables varies types of collaboration among stakeholders. 
In particular, we consider a multi-hop network with arbitrary topology, where the nodes collaboratively finish multiple computation tasks. Nodes have heterogeneous computation capabilities and some are also data sources (sensors and mobile users) that generate data for computation tasks. Each task has a requester node for the computation result. 
We allow partial offloading, where the computation of a task can be partitioned and conducted by multiple nodes, by considering 
data-driven and
divisible tasks, i.e., 
a task can be divided arbitrarily into sub-tasks that compute different parts of the input data.
Finishing a task requires the routing of data from possibly multiple data sources to multiple nodes for computation and the routing of results to the requester. We aim for a joint routing (how to route the data/result) and computation offloading (where to compute) strategy for all tasks that minimizes the total communication and computation costs.
This problem is non-trivial, but we are able to find the optimal joint routing and computation strategy in the proposed framework: In particular, we provide the necessary and sufficient conditions for the global optimum and propose an distributed and adaptive algorithm that converges to the global optimum with small information exchange overhead.

Compared with existing edge computing and CEC studies, our work is distinct in several aspects. First, 
while most works study only offloading-type computation where requesters offload their data for computation \cite{hong2019multi,liu2020distributed,al2016distributed},
we allow data sources and requesters be different nodes to cover computation that requires fetching data from other nodes (like sensors) besides the requester itself. Sahni et al. \cite{sahni2017edge,sahni2018data} also consider arbitrary data sources, but are restricted to fully connected networks or predefined routing protocols.
Second, while we study routing for both data and result flows, most existing works ignore the cost for transmitting results \cite{luo2021qoe,sahni2020multi,shi2019area} and simply let results transmitted along the reverse path of data \cite{he2021multi,funai2019computational}. However, the size of computation result is not negligible in many applications. For example, the communication cost of intermediate results is one major concern in federated or distributed machine learning \cite{mcmahan2017communication}, and result size can even be larger than data size in applications like decompression and image enhancement.
Third, we model congestion in our framework by considering non-linear communication and comuputation costs instead of linear costs as in \cite{hong2019multi,hong2019qos}. 
Finally, the existence of optimal solution, and distributed and adaptive algorithm makes our framework possible to be implemented in practical distributed edge networks with provable performance guarantee and robustness to network condition changes. 

Our detailed contributions are as follows:
\begin{itemize}
    \item To the best of our knowledge, we formulate the fist framework that jointly studies partial offloading and routing for both data and result in arbitrary CEC networks with congestible links. 
    \item We provide the global optimal routing and offloading strategy for this non-convex problem, by studying the necessary and sufficient optimality condition.
    \item We devise a distributed, asynchronous, and adaptive algorithm that converges to the global optimum. 
    \item By extensive experiments, we show the advantages of the proposed algorithm over several baselines in different network topologies.
\end{itemize}
The remainder of this paper is organized as follows. We present our model of a CEC network and formulate the optimization problem in Section~\ref{Section:Model}. The necessary and sufficient optimality conditions are discussed in Section~\ref{Section: Optimality}, while the proposed algorithm is presented in Section~\ref{Section: Algorithm}. We summarize our experiment results in Section~\ref{Section:Simulation} and finally conclude in Section~\ref{Section:Conclusion}.

\section{Network Model and Problem Formulation} \label{Section:Model}

We begin by presenting our formal model of a collaborative edge computing network where multiple stakeholders collaborate to finish computation tasks. Such network is motivated by several real-word applications such as IoT networks, connected vehicles and UAV swarms.
An example system that involves IoT network on the edge is shown in Fig. \ref{fig1}. 

\subsection{Network model }

We consider a quasi-static network, represented by a directed and strongly connected graph $G = (V,E)$ where $V$ is the set of nodes (devices with either task, data or computation resource to share) and $E$ is the set of links. 
For nodes $i, j \in V$, there exists link $(i,j) \in E$ if a feasible connection from $i$ to $j$ is available for data or result forwarding. Denote $\mathcal{I}(i)$ the set of all incoming nodes $j\in V$ of $i$ such that $(j,i) \in E$, and $\mathcal{O}(i)$ the set of outgoing nodes $j \in V$ of $i$ such that $(i,j) \in E$.

\begin{figure}[htbp]
\centerline{\includegraphics[width=0.35\textwidth]{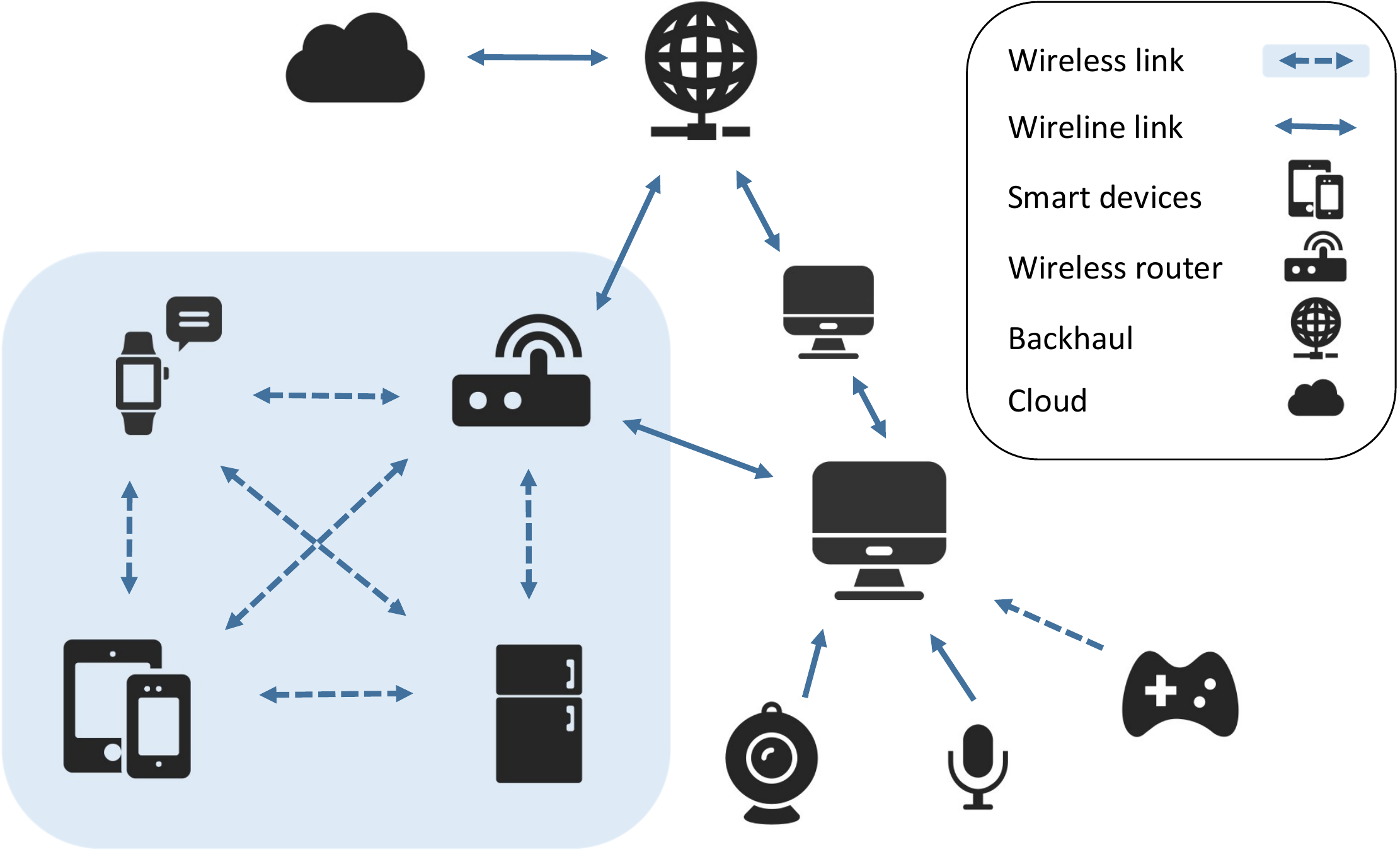}}
\vspace{-0.3\baselineskip}
\caption{Sample system topology involving IoT network on the edge}
\label{fig1}
\vspace{-0.5\baselineskip}
\end{figure}

We consider general computation tasks that map input data to result of non-negligible size, including image/video compression, message encoding/decoding, model training\footnote{Model training usually involves multiple training epochs, in each the gradient of the loss function for every data point in the dataset is computed.}, etc.
A task involves the computation of data generated by multiple data sources and the computation results being forwarded to a given destination (the requester).
Specifically, we assume the computation performed in the network are of $M$ different types, and denote $S$ the set of all task. A task is represented by a pair $(d,m) \in S$, where $d \in V$ is the result destination and $m \in [M]$ is a specified computation type\footnote{We denote $[M]$ the set of integers $\{1, 2,\dots,M\}$ throughout the paper.}. Let $r_i(d,m) \geq 0$ be the exogenous input rates of the data corresponding to task $(d,m)$ at data source node $i \in V$.
With this model we enable multi-source-single-destination tasks, namely multiple nodes $i$ with $r_i(d,m)>0$ are allowed for a task $(d,m)$, e.g., a monitor process involving multiple sensors. The input rate $r_d(d,m)$ at destination could also be positive, incorporating applications with locally provided data and computation offloading.



\begin{table}[t!]
\footnotesize
\begin{center}
\begin{tabular}{l | l }
\hline
$G = (V,E)$ & Network graph $G$, set of nodes $V$ and links $E$\\
$M$ & Total number of computation types\\
$t = (d,m)$ & Task $t$, destination $d$ and computation type $m$ \\
$S$ & Set of all tasks\\
$r_i(d,m)$ & Input data rate of task $(d,m)$ at node $i$ \\
$\mathcal{O}(i)$, $\mathcal{I}(i)$ & Out-neighbors and in-neighbors of node $i$\\
$f^-_{ij}(d,m)$ & Data flow of task $(d,m)$ on link $(i,j)$\\
$f^+_{ij}(d,m)$ & Result flow of task $(d,m)$ on link $(i,j)$\\
$g_i(d,m)$ & Flow assigned to computation of $(d,m)$ at $i$\\
$a_m$ & Result size per unit input data of computation $m$ \\
$t^-_{i}(d,m)$ & Total data flow of $(d,m)$ at $i$ \\
$t^+_{i}(d,m)$ & Total result flow of $(d,m)$ at $i$ \\
$\phi^-_{ij}(d,m)$ & Fraction of $t^-_{i}(d,m)$ forwarded to node $j$, $j \neq 0$ \\
$\phi^-_{i0}(d,m)$ & Fraction of $t^-_{i}(d,m)$ assigned to computation at $i$ \\
$\phi^+_{ij}(d,m)$ & Fraction of $t^+_{i}(d,m)$ forwarded to node $j$\\
$F_{ij}$ & Total flow on link $(i,j)$\\
$g_i^m$ & Computation input amount of task type $m$ at $i$\\
$\bf{G}_i$ & Vector of computation input for all task types at $i$ \\
$D_{ij}(F_{ij})$ & Communication cost (e.g. queueing delay) on $(i,j)$\\
$C_i(\bf{G}_i)$ & Computation cost (e.g. CPU time) at node $i$ \\
$D$ & Sum of all communication and computation costs\\
\hline
\end{tabular}
\end{center}
\vspace{-1\baselineskip}
\caption{ Major notations used in network model}
\vspace{-1\baselineskip}
\label{table:tab1}
\end{table}

\subsection{Routing and computation strategy}
The injecting data flows of each task are routed to nodes with computation resources to be computed. After computation, result flows are generated and routed to corresponding destination nodes.
Since data and computation result are simultaneously forwarded in the network, we distinguish them by superscript $-$ and $+$, respectively.

\noindent\textbf{Data and result flows.} We consider a hop-by-hop routing scheme: Let $f_{ij}^-(d,m) \geq 0$ denote the data flow of task $(d,m)$ on link $(i,j)$ and $f_{ij}^+(d,m) \geq 0$ denote the result flow on the same link.
Let $g_i(d,m) \geq 0$ denote the computation flow at node $i$ corresponding to task $(d,m)$, i.e., the data flow forwarded to the processor of node $i$ for computation. 
We consider that the
result flow generated by computation
for type $m$ is a non-decreasing function $\gamma_m(\cdot)$ of the corresponding computation flow, given as
In particular, we assume that $\gamma_m(\cdot)$ is a non-negative weighted sum of a linear 
function and a sign function. Let $\mathbbm{1}_{A} = 1$ if statement $A$ is true and $0$ otherwise, then
   $$ \gamma_m(g_i(d,m)) = a_m g_i(d,m) + b \mathbbm{1}_{g_i(d,m)>0}, $$
where $a_m \geq 0$ is the ratio of result size versus data, $b \geq 0$ is a fixed overhead due to 
task partition
or the nature of computation. {In this paper, we focus on case $b = 0$.} Such assumption fits most of the computation required in modern applications with usually $a_m \leq 1$. 
{Whereas we also allow $a_m > 1$, representing special types of computation with result size larger than the input data size, e.g., video rendering, image super-resolution or file decompression.\footnote{We defer the analysis of case $b \neq 0$ to our future work. As an expansion, subject to minor modifications, the main mathematical conclusions of this paper apply for cases where $\gamma_m(\cdot)$ is increasing and convex.}}

Let $t_i^-(d,m) $ denote the total data traffic of task $(d,m)$ forwarded and injected to node $i$, and $t_i^+(d,m)$ denote the total result traffic forwarded to and generated at node $i$, given by the following 
\begin{align}
      t_i^-(d,m) &= \sum\nolimits_{j \in \mathcal{I}(i)}f_{ji}^-(d,m) + r_i(d,m), \label{ti-}
    \\t_i^+(d,m) &= \sum\nolimits_{j \in \mathcal{I}(i)}f_{ji}^+(d,m) + \gamma_m(g_i(d,m)). \label{ti+}
\end{align}

\noindent\textbf{Routing and computation strategy.} To describe the computation and forwarding scheme in a distributed fashion, we assume that each node has two virtual router, one for data flow and one for result flow, as shown in Fig. \ref{fig2}. Let $\phi_{ij}^-(d,m)$, $\phi_{ij}^+(d,m)\in [0,1] $ be the fraction of data or result flow of task $(d,m)$ at node $i$ forwarded to node $j$. 
For a coherent notation, we also use $\phi_{i0}^-(d,m)\in [0,1]$ to denote the fraction of computation flow at node $i$, and use $\boldsymbol{\phi}^- = [\phi_{ij}^-(d,m)]_{i\in V, j\in\{0\}\cup\mathcal{O}(i), (d,m)\in S}$, $\boldsymbol{\phi}^+ = [\phi_{ij}^+(d,m)]_{(i,j)\in E, (d,m)\in S}$ and $\boldsymbol{\phi} = (\boldsymbol{\phi}^-,\boldsymbol{\phi}^+)$ to represent the system-wide 
routing and computation strategy.
To ensure all tasks are fulfilled, all injected data must be computed and the result must be delivered to the destination, given by the following flow conservation: 
for all $(d,m) \in S$ and $i\in V$, the data traffic is either computed or forwarded, 
\begin{align}
    &f_{ij}^-(d,m) = t_i^-(d,m) \phi_{ij}^-(d,m), \label{FlowConservation1.1}
    \\ &g_i(d,m) = t_i^-(d,m) \phi_{i0}^-(d,m), \label{FlowConservation1.2}
    \\ &\sum\nolimits_{j \in \left\{0\right\} \bigcup \mathcal{O}(i) } \phi_{ij}^-(d,m) = 1, \label{FlowConservation1.3}
\end{align}
and for the result traffic, the destination is a sink:
\begin{align}
    & f_{ij}^+(d,m) = t_i^+(d,m) \phi_{ij}^+(d,m), \label{FlowConservation2.1}
    \\&  \sum_{j \in \mathcal{O}(i) } \phi_{ij}^+(d,m) = \begin{cases} 
    1, \quad \text{if } i \neq d, 
    \\ 0, \quad \text{if } i = d.
    \end{cases} 
    \label{FlowConservation2.2}
\end{align}

\begin{figure}
\centerline{
\includegraphics[width=0.4\textwidth]{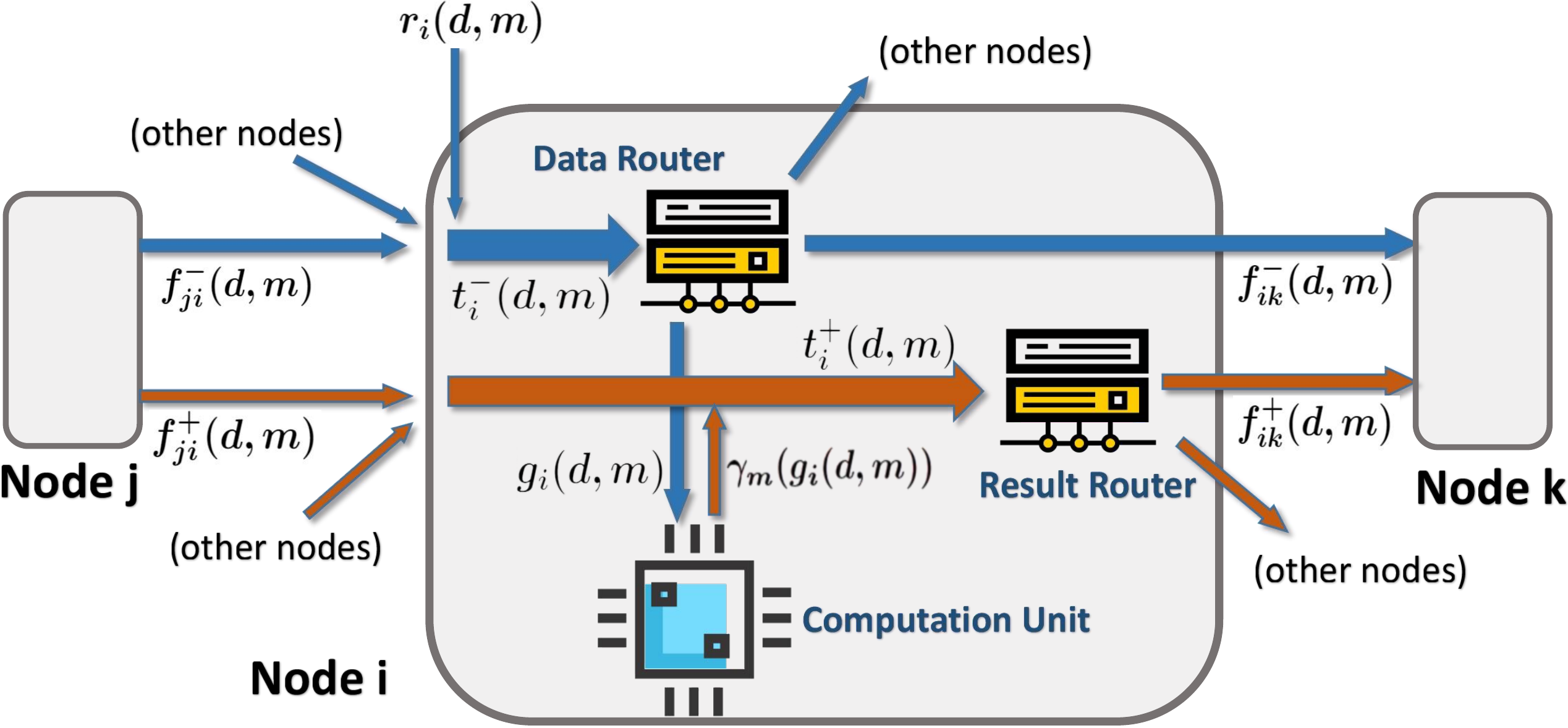}}
\vspace{-0.2\baselineskip}
\caption{Example of data and result flows for nodes $j \to i \to k$ }
\label{fig2}
\vspace{-0.8\baselineskip}
\end{figure}

\subsection{Communication and computation cost}
For mathematical formulation, we 
ignore the request messages sent by the requester to data sources, as the size of such messages are typically negligible compared to data or result, and could be delivered using a separate channel.

Instead of sharp bandwidth or computation capacity constraint in \cite{liu2019joint} or linear costs in \cite{ren2019collaborative}, we assign convex costs for D2D communication and local computation depending on corresponding flow rate, a more general 
assumption reflecting network congestion status.
In particular, define the total flow on link $(i,j) \in E$ as
$$\textstyle F_{ij} = \sum_{(d,m) \in S} \left(f_{ij}^+(d,m) + f_{ij}^-(d,m)\right),$$ 
 we assume the communication cost on $(i,j)$ is $D_{ij}(F_{ij})$, where $D_{ij}(\cdot)$ is increasing, continuously differentiable and convex. 
Similarly, denote $\boldsymbol{G}_i$ the data rate vector of the amount of computation for each type performed at node $i$, i.e.,
$$ \textstyle \boldsymbol{G}_i = \left[g_i^1, g_i^2, \cdots, g_i^M\right],$$
where $g_i^m$ is the total computation amount of type $m$, 
$$\textstyle g_i^m = \sum_{d:(d,m) \in S}g_i(d,m).$$ 
We denote $C_i(\boldsymbol{G}_i)$ the cost for node $i$ to fulfill the computation load $\boldsymbol{G}_i$, where function $C_i(\cdot)$ is an increasing, continuously differentiable and convex multivariable function, i.e., increasing on every coordinate and jointly convex in $g_i^1, \cdots, g_i^M$.

{ 
Such assumption of cost functions incorporates a variety of existing performance metrics.
For example, $D_{ij}(F_{ij}) = {F_{ij}}/\left({\mu_{ij}-F_{ij}}\right)$ gives the average number of packets waiting for or under transmission at link $(i,j)$, provided that $\mu_{ij}$ is the service rate in an M/M/1 queue model\cite{bertsekas2021data} and $F_{ij} < \mu_{ij}$. 
One could also approximate the sharp capacity constraint $F_{ij} \leq C_{ij}$ (e.g., in \cite{liu2019joint}) by a smooth convex function that goes to infinity when approaching to capacity limit $C_{ij}$. 

Note that we measure the computation cost $C_i(\boldsymbol{G}_i)$ as a function of data rate, which is adopted in network function virtualization (NFV) studies\cite{zhang2018optimal}, but most previous computation offloading researches \cite{hong2019multi}\cite{chen2021edgeconomics} assign data amount and computation workload separately for a task. 
Our formulation can be reduced to the latter if choose certain cost function. 
For instance, consider the scenario where users make requests with input data size $d_m$ (bits) and computation workload $c_m$ (CPU cycles). 
By setting $C_i(\boldsymbol{G_i}) = \sum_{m \in [M]} c_m g_i^m / d_m$, we measure CPU cycles as cost. Or by setting $v_i^m$ to be the computation speed of type $m$ at $i$ and $C_i(\boldsymbol{G_i}) = \sum_{m \in [M]} c_m g_i^m / (d_m v_i^m)$, we measure CPU runtime as cost.


Note that for a network with heterogeneous computation resources, our formulation is even more flexible than that in \cite{sahni2020multi}\cite{chen2021edgeconomics}, where there a task always assigns the same computation workload wherever it is computed. In fact, our model captures the fact that in practical network edge, the workload for a certain task may be very different depending on where to perform it, e.g., some parallelizable computation is easier at nodes equipping GPU, but slower at others.
}

\subsection{Joint routing and computation offloading problem}
  In this paper, we aim at minimizing the overall cost of edges and devices for both communication and computation,
\begin{align}
    \min_{\boldsymbol{\phi}} \quad & D = \sum_{(i,j) \in E}D_{ij}(F_{ij}) + \sum_{i \in V} C_i(\boldsymbol{G}_i) \label{JointProblem_nodebased}
    \\ \text{such that} \quad & \text{(\ref{ti-}) to (\ref{FlowConservation2.2}) hold} \nonumber
\end{align}
Note that problem (\ref{JointProblem_nodebased}) is not convex in $\boldsymbol{\phi}$. We will demonstrate by example that standard gradient-based methods solving for KKT points highly likely generate sub-optimal solutions. 


\section{Sufficient Optimality Condition} \label{Section: Optimality}

In this section, we provide necessary and sufficient conditions for the global optimum of problem (\ref{JointProblem_nodebased}). 
Our analysis follows \cite{gallager1977minimum}, while we make non-trivial extensions for considering both data and result flows, as well as in-network computation.
Remind that we focus on the case with $b = 0$, where the size of computation result is in proportion to the data size, i.e., $\gamma_m(g_i(d,m)) = a_m g_i(d,m)$. 

We start by giving closed-form expressions to the derivatives of the total cost $D$. 
For an increment of exogenous data flow $r_i(d,m)$, the increase of $D$ (i.e. the marginal cost) is caused by two aspects, (1) the cost of forwarding extra data flow to $i$'s outgoing neighbors $j\in\mathcal{O}(j)$, and (2) the cost of assigning extra computation load at $i$'s computation unit. Note that the first aspect 
could be further decomposed into two terms, the extra cost on the out-link $(i,j)$  and the extra cost at the next-hop $j$, and similarly could the second aspect.
Thus formally, the marginal cost at node $i$ is given as
\begin{equation}
\begin{aligned}
    \frac{\partial D}{ \partial r_i(d,m)} &= \sum_{j \in \mathcal{O}(i)} \phi_{ij}^-(d,m) \left[D_{ij}^\prime(F_{ij}) + \frac{\partial D}{ \partial r_j(d,m)}\right]
    \\ &+ \phi_{i0}^-(d,m) \left[\frac{\partial C_i(\boldsymbol{G_i})}{\partial g_i^m} + a_m \frac{\partial D}{\partial t_i^+(d,m)}\right],
    \label{partial_D_r}
\end{aligned}
\end{equation}
where ${\partial D}/{\partial t_i^+(d,m)}$ denotes the marginal cost corresponding to an increment of result traffic at $i$. Similarly to (\ref{partial_D_r}), this marginal of the result traffic is a weighted sum of extra costs at out-links and at next-hope nodes, given as 
\begin{align}
    \frac{\partial D}{\partial t_i^+(d,m)} = \sum_{j \in \mathcal{O}(i)} \phi_{ij}^+(d,m) \left[D_{ij}^\prime(F_{ij}) +  \frac{\partial D}{\partial t_j^+(d,m)}\right].
    \label{partial_D_t}
\end{align}
Note that \eqref{partial_D_r} and \eqref{partial_D_t} can be calculated recursively, whereas we defer the detailed mechanism to Section \ref{Section: Algorithm}.
Meanwhile, with an increment of $\phi^-_{ij}$ or $\phi^+_{ij}$, the extra cost could also be decomposed in the similar way as in \eqref{partial_D_r} or \eqref{partial_D_t}, respectively. 
\begin{align}
    \frac{\partial D}{\partial \phi_{ij}^-(d,m)} &= 
    \begin{cases}
    t_i^-(d,m) \left[D_{ij}^\prime(F_{ij}) + \frac{\partial D}{\partial r_j(d,m)}\right], \text{if } j \neq 0
    \\ t_i^-(d,m) \left[ \frac{\partial C_i(\boldsymbol{G_i})}{\partial g_i^m} + a_m \frac{\partial D}{\partial t_i^+(d,m)} \right], \text{ if } j =0
    \end{cases} \label{partial_D_phi-}
    \\ \frac{\partial D}{\partial \phi_{ij}^+(d,m)} &= 
     t_i^+(d,m) \left[D_{ij}^\prime(F_{ij}) + \frac{\partial D}{\partial t_j^+(d,m)} \right].\label{partial_D_phi+}
\end{align}

Then we could solve (\ref{JointProblem_nodebased}) by minimizing the Lagrangian
\begin{equation}
\begin{aligned}
    & L(\boldsymbol{\phi},\lambda^-,\lambda^+) = D 
    \\ &- \sum_{i \in V} \sum_{(d,m) \in S} \lambda_{i,(d,m)}^-(\sum_{j \in \left\{ 0\right\} \cup \mathcal{O}(i) } \phi_{ij}^-(d,m) - 1) 
    \\ &- \sum_{i \in V} \sum_{(d,m) \in S}\lambda_{i,(d,m)}^+(\sum_{j \in \mathcal{O}(i) } \phi_{ij}^+(d,m) - \mathbbm{1}_{i \neq d})
\end{aligned}    
\label{Lagrangian}
\end{equation}
subject to the constraints $\phi_{ij}^-(d,m) \geq 0$ and $\phi_{ij}^+(d,m) \geq 0$.
By setting the derivative of $L$ to $0$, a KKT necessary condition of the global minimizer to (\ref{JointProblem_nodebased}) is given by Lemma \ref{Lemma_Necessary}.

\begin{lemma} \label{Lemma_Necessary}
Let $b = 0$ and $\boldsymbol{\phi}^-,\boldsymbol{\phi}^+$ be the global solution that minimizes (\ref{JointProblem_nodebased}), then for all $i$ and $(d,m)$, and for all $j\in\left\{0\right\} \cup \mathcal{O}(i)$ w.r.t. data flow or  $j\in\mathcal{O}(i)$ w.r.t. result flow,
\begin{align*}
    &\frac{\partial D}{ \partial \phi_{ij}^-(d,m)}  
    \begin{cases}
     = \min\limits_{k \in \left\{0\right\} \cup \mathcal{O}(i) } \frac{\partial D}{ \partial \phi_{ik}^-(d,m)} , \quad \text{if } \phi_{ij}^-(d,m) >0,
    \\ \geq \min\limits_{k \in \left\{0\right\} \cup \mathcal{O}(i) } \frac{\partial D}{ \partial \phi_{ik}^-(d,m)} ,  \quad \text{if } \phi_{ij}^-(d,m) =0,
    \end{cases} 
   \\ &\frac{\partial D}{ \partial \phi_{ij}^+(d,m)}  
     \begin{cases}
     = \min\limits_{k \in \mathcal{O}(i) } \frac{\partial D}{ \partial \phi_{ik}^+(d,m)} , \quad \text{if } \phi_{ij}^+(d,m) >0,
    \\ \geq \min\limits_{k \in \mathcal{O}(i) } \frac{\partial D}{ \partial \phi_{ik}^+(d,m)}, \quad \text{if } \phi_{ij}^+(d,m) =0.
    \end{cases} 
\end{align*}
\end{lemma}

Note that the condition in Lemma \ref{Lemma_Necessary} is not a sufficient condition for optimality. A toy example for such non-sufficiency is provided in Fig.\ref{fig3}:
The only task is $(4,m)$ with input data only occur at node $1$,
the routing/offloading strategy and marginal costs are shown on figure. 
It is easy to verify that condition in Lemma \ref{Lemma_Necessary} is satisfied.
However, by increasing $\phi^-_{24}$ and decreasing $\phi^-_{23}$, the input marginal ${\partial D}/{\partial r_2}$ will decrease and thus ${\partial D}/{\partial \phi^-_{12}}$ will decrease. In this case, the objective $D$ could be improved by increasing $\phi^-_{12}$ and decreasing $\phi^-_{14}$.

\begin{figure}[htbp]
\centerline{\includegraphics[width=0.4\textwidth]{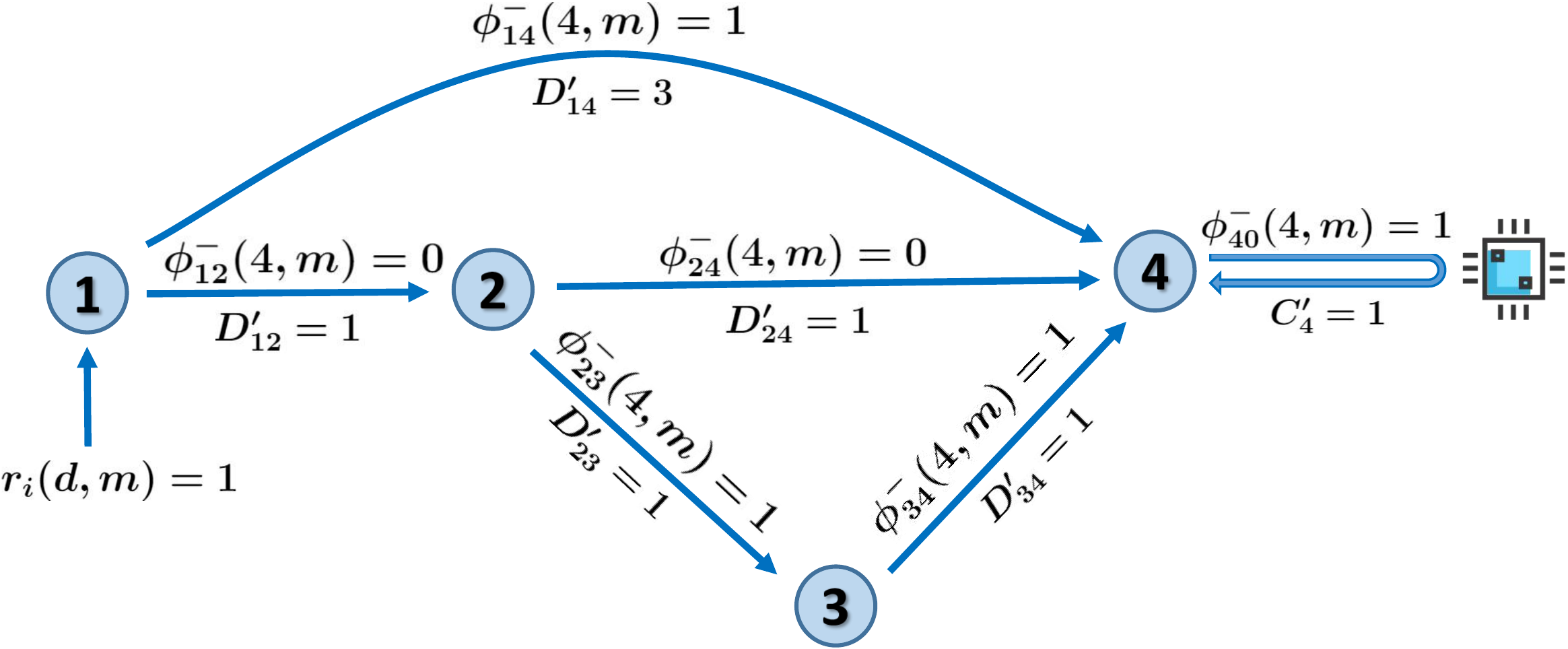}}
\vspace{-0.2\baselineskip}
\caption{A non-optimal situation satisfying the condition in Lemma \ref{Lemma_Necessary}}
\label{fig3}
\vspace{-0.5\baselineskip}
\end{figure}

The underlying intuition for such non-sufficiency is that the condition in Lemma \ref{Lemma_Necessary} automatically holds if $t_i^{-}(d,m) = 0$ and $t_i^{+}(d,m) = 0$, no matter what the routing/computation strategy is. 
Nevertheless, given $t_i^-(d,m)$ and $t_i^+(d,m)$ exist identically in \eqref{partial_D_phi-} and \eqref{partial_D_phi+} respectively for all $j$, we remove them and devise an augmented condition specified in Theorem \ref{Thm_Sufficient}, which instead is shown to be sufficient for global optimality.

\begin{theorem} \label{Thm_Sufficient}
Let $b = 0$ and $\boldsymbol{\phi}^-,\boldsymbol{\phi}^+$ be feasible to (\ref{JointProblem_nodebased}), if the following holds: for all $i$ and $(d,m)$, and for all $j\in\left\{0\right\} \cup \mathcal{O}(i)$ w.r.t data flow or $j\in\mathcal{O}(i)$ w.r.t. result flow,
\begin{align*}
    &\delta_{ij}^-(d,m) \begin{cases}
    = \min\limits_{k \in \left\{0\right\}  \cup \mathcal{O}(i)} \delta_{ik}^-(d,m), \quad \text{if } \phi_{ij}^-(d,m) >0
    \\ \geq \min\limits_{k \in \left\{0\right\}  \cup \mathcal{O}(i)} \delta_{ik}^-(d,m), \quad \text{if } \phi_{ij}^-(d,m) =0
    \end{cases}
    \\& \delta_{ij}^+(d,m) \begin{cases}
    = \min\limits_{k \in \mathcal{O}(i)} \delta_{ik}^+(d,m), \quad \text{if } \phi_{ij}^+(d,m) >0
    \\ \geq \min\limits_{k \in \mathcal{O}(i)} \delta_{ik}^+(d,m),  \quad \text{if } \phi_{ij}^+(d,m) =0
    \end{cases}
\end{align*}
where the augmented marginals $\delta_{ij}^-(d,m)$ and $\delta_{ij}^+(d,m)$ are
\begin{equation}
\begin{aligned}
    \delta_{ij}^-(d,m) &= 
    \begin{cases}
    D_{ij}^\prime(F_{ij}) + \frac{\partial D}{\partial r_j(d,m)}, &\quad \text{if } j \in \mathcal{O}(i)
    \\ \frac{\partial C_i(\boldsymbol{G_i})}{\partial g_i^m} + a_m  \frac{\partial D}{\partial t_i^+(d,m)} , &\quad \text{if } j =0
    \end{cases} 
    \\ \delta_{ij}^+(d,m) &=  D_{ij}^\prime(F_{ij}) +  \frac{\partial D}{\partial t_j^+(d,m)},
\end{aligned}\label{delta}
\end{equation}
then $(\boldsymbol{\phi}^-,\boldsymbol{\phi}^+)$ is a global optimal solution to (\ref{JointProblem_nodebased}).
\end{theorem}
\begin{proof}
{See Appendix}.
\end{proof}

Theorem \ref{Thm_Sufficient} is the main theoretical result of this paper, and in fact a practical criterion for algorithm implementation. As a simple illustration of the difference between Theorem \ref{Thm_Sufficient} and Lemma \ref{Lemma_Necessary}, we further assume the network in Fig. \ref{fig3} has linear communication costs. 
It turns out that for any routing scheme satisfying Theorem \ref{Thm_Sufficient}, we must have $\phi^-_{12}(4,m) = 1$ and $\phi^-_{24}(4,m) = 1$, which precisely indicates the shortest path $1 \to 2 \to 4$ for data flow.

\section{Distributed and Adaptive Algorithm}
\label{Section: Algorithm}

In this section, we introduce a distributed and adaptive algorithm that converges to the global optimal solution of \eqref{JointProblem_nodebased} specified by Theorem \ref{Thm_Sufficient}, based on scaled gradient projection. 
We allow nodes update their routing-computation strategies in an autonomous and asynchronous manner, and adapt to the changes of input rate and network topology.
Our method follows Xi and Yeh\cite{xi2008node}, and further distinguishes data and result flows by improving control messages exchanging protocol.

\subsection{Algorithm overview and loop-free property}
We first introduce the \emph{loop-free} property of a global strategy $\boldsymbol{\phi}$. 
For a task $(d,m)$, there is a \emph{data path} from node $i$ to node $j$ ($j \neq i$) if there is a sequence of node $n_1, \cdots ,n_L$ such that $(n_l,n_{l+1}) \in E$ and $\phi^-_{n_l n_{l+1}}(d,m) > 0$ for $l = 1,\cdots,L-1$, with $n_1 = i$ and $n_L = j$, where $L$ is the \emph{hop number} for this data path. We say $\boldsymbol{\phi}$ has a \emph{data loop} if there exists task $(d,m)$ and node $i$, $j$ such that $i$ has a data path to $j$, and 
vice versa.
Similarly, we can define \emph{result path} and \emph{result loop} for result flows.
Then, we say strategy $\boldsymbol{\phi}$ is \emph{loop-free} if it has neither data loop nor result loop.
{
Loop-free is a fundamental requirement for $\boldsymbol{\phi}$ to be feasible and is guaranteed throughout the algorithm. Because given a slight increase in input rate, a loop may cause the on-loop flow to build up to infinity.} \footnote{We allow loops concatenated by a data path and a result path of the same task, which occurs in scenarios where the destination is the data source.}  

We denote respectively by $\boldsymbol{\phi}^-_{i}(d,m)$ and $\boldsymbol{\phi}^+_{i}(d,m)$ node $i$'s strategy vector $\left[\phi^-_{ij}(d,m)\right]_{j \in \left\{0\right\} \cup \mathcal{O}(i)}$ and $\left[\phi^+_{ij}(d,m)\right]_{j \in \mathcal{O}(i)}$, 
and by $\boldsymbol{\delta}^-_{i}(d,m)$ and $\boldsymbol{\delta}^+_{i}(d,m)$ the vectors of augmented marginals (defined in \eqref{delta}), $\left[\delta^-_{ij}(d,m)\right]_{j \in \left\{0\right\} \cup \mathcal{O}(i)}$ and $\left[\delta^+_{ij}(d,m)\right]_{j \in \mathcal{O}(i)}$.
We assume the network starts with a feasible and loop-free state ${\boldsymbol{\phi}}^0$. At $t$-th iteration, each node $i$ updates its strategy corresponding to task $(d,m)$ with the following scaled gradient projection variant
\begin{equation}
\begin{aligned}
    &\boldsymbol{\phi}^-_i(d,m)^{t+1} = \underset{\boldsymbol{v} \in \mathcal{D}_i^-(d,m)^t}{\arg\min} \boldsymbol{\delta}^-_{i}(d,m)^t \cdot (\boldsymbol{v} - \boldsymbol{\phi}^-_i(d,m)^{t}) 
    \\&\quad+ (\boldsymbol{v} - \boldsymbol{\phi}^-_i(d,m)^{t})^T M^-_i(d,m)^t (\boldsymbol{v} - \boldsymbol{\phi}^-_i(d,m)^{t}),
\end{aligned} \label{variable_update}    
\end{equation}
and $\boldsymbol{\phi}^+_i(d,m)^{t+1}$ is updated similarly with ``$-$'' replaced by ``$+$'', where $M^-_i(d,m)^t$, $M^+_i(d,m)^t$ are symmetric and positive semi-definite scaling matrices 
designed to achieve good convergence properties.
$\mathcal{D}_i^-(d,m)^t$ is the feasible set of $\boldsymbol{\phi}^-_i(d,m)^{t+1}$ 
given by $\boldsymbol{\phi}_i \geq \boldsymbol{0}$ and 
\begin{align*}
\sum\nolimits_{ j \in \{0\} \cup \mathcal{O}(i)}  \phi_{ij} = 1; \quad \phi_{ij} = 0 \text{ for } \forall j \in \mathcal{B}_i^-(d,m)^t ,
\end{align*}
and feasible set $\mathcal{D}_i^+(d,m)^t$ is defined similarly, where $\mathcal{B}_i^-(d,m)^t$ and $\mathcal{B}_i^+(d,m)^t \subseteq V$ are the \emph{blocked nodes} of $i$ relevant to data and result of task $(d,m)$ to guarantee the feasibility and loop-free property. 
We next describe in detail the estimation of marginals, and how to obtain matrices $M^-_i(d,m)^t$, $M^+_i(d,m)^t$ and sets $\mathcal{B}_i^-(d,m)^t$, $\mathcal{B}_i^+(d,m)^t$. 
{We emphasis that our proposed algorithm is not pure gradient-based, as the gradients are replaced by the augmented marginals $\boldsymbol{\delta}^-_{i}(d,m)$ and $\boldsymbol{\delta}^+_{i}(d,m)$, corresponding to Theorem \ref{Thm_Sufficient}.}

\subsection{Marginal cost estimation by broadcast}
\label{subsection: broadcast}
Each node $i$ needs to compute its augmented marginal cost vectors $\boldsymbol{\delta}^-_{i}(d,m)$ and $\boldsymbol{\delta}^+_{i}(d,m)$ following \eqref{delta}. 
Node $i$ can directly estimate $D^\prime_{ij}(F_{ij})$ and $C^\prime_{i}(\boldsymbol{G_i})$ while sending and receiving any message on link $(i,j)$ 
or performing computation at local computation unit. 
To obtain $\frac{\partial D}{\partial r_i(d,m)}$ and $\frac{\partial D}{\partial t^+_i(d,m)}$, we introduce a two-stage broadcast protocol:
 
  \textbf{1)}  To calculate $\frac{\partial D}{\partial t^+_i(d,m)}$, node $i$ first waits until receives messages carrying $\frac{\partial D}{\partial t^+_j(d,m)}$ from 
    all downstream nodes 
    $j\in\mathcal{O}(i)$
    such that $\phi^+_{ij}(d,m) > 0$
    , and estimated $D^\prime_{ij}(F_{ij})$. Then calculates its own $\frac{\partial D}{\partial t^+_i(d,m)}$ according to (\ref{partial_D_t}) and broadcasts this to 
    all upstream nodes 
    $k\in\mathcal{I}(i)$ such that $\phi^+_{ki}(d,m) > 0$. 

\textbf{2)}  For the $\frac{\partial D}{\partial r_i(d,m)}$, a similar procedure is used to compute and broadcast $\frac{\partial D}{\partial r_i(d,m)}$ according to (\ref{partial_D_r}). Note that node $i$ must obtain $\frac{\partial D}{\partial t^+_i(d,m)}$ before calculating and broadcasting $\frac{\partial D}{\partial r_i(d,m)}$. 

With the loop-free property held, such broadcast starting at destination $d$ (with $\frac{\partial D}{\partial t^+_d(d,m)} = 0$) is guaranteed to traverse throughout the network. 
Though data flows have no explicit (and fixed) sink, the loop-free property guarantees the broadcast of stage 2) could successfully start with the last node of each data path.

\subsection{Blocked nodes and scaling matrices}
\label{subsection: Blocked nodes and scale matrices}
To achieve the feasibility and loop-free property, we consider sets $\mathcal{B}_i^-(d,m)^t, \mathcal{B}_i^+(d,m)^t \subseteq V$, to nodes in which node $i$ is forbidden to forward data or result of task $(d,m)$, respectively. 
By Theorem \ref{Thm_Sufficient} combined with expression \eqref{partial_D_r} and \eqref{partial_D_t}, at a global optimal strategy, the input marginals $\frac{\partial D}{\partial r_i(d,m)}$ or $\frac{\partial D}{\partial t^+_i(d,m)}$ should be monotonically decreasing along any data-path or result-path. We thus mandate that node $i$ should not increase flow rate to a neighbor $j$ that either (1) has higher input marginal, or (2) could form a data/result-path containing some link $(p,q)$ and $q$ has higher input marginal than $p$.  We denote by $\mathcal{B}_i^-(d,m)^t, \mathcal{B}_i^+(d,m)^t$ the sets of such neighbor $j$. 
Note that  $\mathcal{B}_i^-(d,m)^t$ and $\mathcal{B}_i^+(d,m)^t$ are defined separately according to $\frac{\partial D}{\partial r_i(d,m)}$ and $\frac{\partial D}{\partial t^+_i(d,m)}$.
Then, the loop-free property is maintained throughout the algorithm if sets of block nodes are practiced in each iteration. 
The feasibility is also guaranteed since when current state approaches a link/processor capacity, the corresponding marginal will grow to infinity, preventing any flow increase. 
The readers are referred to \cite{gallager1977minimum} for detail. 

The scaling matrices $M^-_i(d,m)^t$ and $M^+_i(d,m)^t$ are introduced to improve the convergence speed while guaranteeing convergence
from arbitrary initial points \cite{xi2008node}. Specifically,
\begin{align*}
    &M_i^+(d,m)^t = \frac{t_i^+(d,m)^t}{2} \times \text{diag}\{A_{ij}(D^0) + 
    \\ &\left| \mathcal{O}(i) \backslash \mathcal{B}_i^+(d,m)^t\right| h_j^+(d,m)^t A(D^0) \}_{j \in \mathcal{O}(i) \backslash \mathcal{B}_i^+(d,m)^t},
\end{align*}
where $D^0$ is the overall cost at initial state, $h_j^+(d,m)^t$ is the maximum hop number among all existing result paths from $j$ to destination $d$, operator $diag$ 
forms a diagonal matrix, and
\begin{align*}
    A_{ij}(D^0) &= \sup_{D<D^0}D^{\prime \prime}_{ij}(F_{ij}), \quad 
     A(D^0) &= \max_{(i,j) \in E} A_{ij}(D^0).
\end{align*}
The definition of $M_i^-(d,m)^t$ is almost a repetition as above, but in terms of the data flow.

\subsection{Asynchronous convergence and complexity}


Our algorithm allows nodes to update their variables asynchronously, or with a non-perfect synchronization due to practical constraints such as the broadcast delay in a large-scale network. To formulate this asynchrony, we assume that at $t$-th iteration, only one node $i$ updates either its $\boldsymbol{\phi}_i^-(d,m)$ or $\boldsymbol{\phi}_i^+(d,m)$ for one task $(d,m)\in S$, and let 
\begin{align*}
    T_{\boldsymbol{\phi}^-_i(d,m)} &= \left\{t \big| \text{ node } i \text{ update its } \boldsymbol{\phi}_i^-(d,m) \text{ at iteration } t \right\},
\end{align*}
and similarly as $T_{\boldsymbol{\phi}^+_i(d,m)}$, then Theorem \ref{thm:convergence} holds.

\begin{theorem}
\label{thm:convergence}
Assume $b = 0$ and the network is static, if 
\begin{align*}
    \lim_{t \to \infty} \left|T_{\boldsymbol{\phi}^-_i(d,m)}\right| = \infty, \quad \lim_{t \to \infty} \left|T_{\boldsymbol{\phi}^+_i(d,m)}\right| = \infty,
\end{align*}
then the constructed sequence $\left\{{\boldsymbol{\phi}}^t\right\}_{t\to\infty}$ converges to a point ${\boldsymbol{\phi}}^*$, where ${\boldsymbol{\phi}}^*$ is feasible and loop-free, and the condition in Theorem \ref{Thm_Sufficient} holds (proof see \cite{xi2008node}).
\end{theorem}

We assume that the variables of all nodes are updated one round every time slot of duration $T$, and every broadcast message described in Section \ref{subsection: broadcast} is sent once in a slot. There are $2|E|$ transmissions of control messages corresponding to a task in one slot, and thus totally $2|S||E|$ transmissions, with on average $2|S|/T$ per link/second and at most $2\Bar{d}|S|$ for each node, where $\Bar{d}$ is the largest out-degree, and a singe broadcast message has $O(1)$ size. 
Moreover, let $\Bar{h}$ be the maximum path hop, and $t_c$ be the maximum time for control message transmission, the broadcast procedure yields a delay of at most $2\Bar{h}t_c$.
The variable size for individual node optimization problem is at most $2\Bar{d}|S|$, where each problem, although is a scaled gradient projection, could be efficiently solved by various commercial solvers since the scaling matrix is PSD and diagonal, and the constraint set is simplex. 

\section{Numerical Evaluation} \label{Section:Simulation}

In this section, we evaluate the scaled gradient projection algorithm, i.e., \textbf{SGP} proposed in Section \ref{Section: Algorithm} by simulation. 
We implement several baseline algorithms and compare the performance of those 
against \texttt{SGP} over different networks and parameter settings. Note that we set $b = 0$ in all experiments.

\begin{table}[htbp]
\footnotesize
\begin{tabular}{|c|p{0.01\textwidth}|p{0.01\textwidth}|p{0.01\textwidth}|p{0.01\textwidth}|c|p{0.01\textwidth}|c|p{0.015\textwidth}|}
\hline
\textbf{Network}&\multicolumn{8}{|c|}{\textbf{Parameters}} \\

\textbf{Topology} & \multicolumn{1}{p{0.01\textwidth}}{$|V|$} & \multicolumn{1}{p{0.01\textwidth}}{$|E|$} & \multicolumn{1}{p{0.01\textwidth}}{$|S|$} & \multicolumn{1}{p{0.01\textwidth}}{$|\mathcal{R}|$} & \multicolumn{1}{c}{\textbf{Link}} & \multicolumn{1}{p{0.015\textwidth}}{$\Bar{d}_{ij}$} &\multicolumn{1}{c}{\textbf{Comp}}& $\Bar{s}_i$ \\
\hline
Connected-ER& \multicolumn{1}{p{0.01\textwidth}}{$20$} & \multicolumn{1}{p{0.01\textwidth}}{$40$} & \multicolumn{1}{p{0.01\textwidth}}{$15$} & \multicolumn{1}{p{0.01\textwidth}}{$5$} & \multicolumn{1}{c}{Queue} & \multicolumn{1}{p{0.015\textwidth}}{$10$} &\multicolumn{1}{c}{Sum-Queue}& $12$ \\
Balanced-tree & \multicolumn{1}{p{0.01\textwidth}}{$15$} & \multicolumn{1}{p{0.01\textwidth}}{$14$} & \multicolumn{1}{p{0.01\textwidth}}{$20$} & \multicolumn{1}{p{0.01\textwidth}}{$5$} & \multicolumn{1}{c}{Queue} & \multicolumn{1}{p{0.015\textwidth}}{$20$} &\multicolumn{1}{c}{Sum-Queue}& $15$ \\
Fog & \multicolumn{1}{p{0.01\textwidth}}{$19$} & \multicolumn{1}{p{0.01\textwidth}}{$30$} & \multicolumn{1}{p{0.01\textwidth}}{$30$} & \multicolumn{1}{p{0.01\textwidth}}{$5$} & \multicolumn{1}{c}{Queue} & \multicolumn{1}{p{0.015\textwidth}}{$20$} &\multicolumn{1}{c}{Sum-Queue}& $17$ \\
Abilene & \multicolumn{1}{p{0.01\textwidth}}{$11$} & \multicolumn{1}{p{0.01\textwidth}}{$14$} & \multicolumn{1}{p{0.01\textwidth}}{$10$} & \multicolumn{1}{p{0.01\textwidth}}{$3$} & \multicolumn{1}{c}{Queue} & \multicolumn{1}{p{0.015\textwidth}}{$15$} &\multicolumn{1}{c}{Sum-Queue}& $10$ \\
LHC & \multicolumn{1}{p{0.01\textwidth}}{$16$} & \multicolumn{1}{p{0.01\textwidth}}{$31$} & \multicolumn{1}{p{0.01\textwidth}}{$30$} & \multicolumn{1}{p{0.01\textwidth}}{$5$} & \multicolumn{1}{c}{Queue} & \multicolumn{1}{p{0.015\textwidth}}{$15$} &\multicolumn{1}{c}{Sum-Queue}& $15$ \\
GEANT & \multicolumn{1}{p{0.01\textwidth}}{$22$} & \multicolumn{1}{p{0.01\textwidth}}{$33$} & \multicolumn{1}{p{0.01\textwidth}}{$40$} & \multicolumn{1}{p{0.01\textwidth}}{$7$} & \multicolumn{1}{c}{Queue} & \multicolumn{1}{p{0.015\textwidth}}{$25$} &\multicolumn{1}{c}{Sum-Queue}& $20$ \\
SW & \multicolumn{1}{p{0.01\textwidth}}{$100$} & \multicolumn{1}{p{0.01\textwidth}}{$320$} & \multicolumn{1}{p{0.01\textwidth}}{$120$} & \multicolumn{1}{p{0.01\textwidth}}{$10$} & \multicolumn{1}{c}{(both)} & \multicolumn{1}{p{0.015\textwidth}}{$20$} &\multicolumn{1}{c}{(both)}& $20$ \\

\hline
\textbf{Other}&\multicolumn{8}{|c|}{
$d_{\max} = 30$, $d_{\min} = 10$, $s_{\max} = 30$, $s_{\min} = 2$
} \\

\textbf{Parameters}&\multicolumn{8}{|c|}{
$M = 5$, $r_{\min} = 0.5$, $r_{\max} = 1.5$
}\\
\hline
\end{tabular}
\vspace{-0.2\baselineskip}
\caption{Simulated Network Scenarios}
\label{tab_scenario}
\vspace{-0.5\baselineskip}
\end{table}

We summarize the simulation scenarios in Table \ref{tab_scenario}. 
We evaluate the algorithms in
the following different network topologies: 
\textbf{Connected-ER} is a connectivity-guaranteed Erdős–Rényi graph, generated by uniformly-randomly creating links with probability $p = 0.1$ on a linear network concatenating all nodes.
\textbf{Balanced-tree} is a complete binary tree.
\textbf{Fog} is a sample topology for fog-computing, where nodes on the same layer are linearly linked in a balance tree \cite{kamran2019deco}.
\textbf{Abilene} is the topology of the predecessor of \emph{Internet2 Network} \cite{rossi2011caching}.
\textbf{GEANT} is a pan-European data network for the research and education community \cite{rossi2011caching}.
\textbf{SW} (small-world) is a ring-like graph with additional short-range and long-range edges \cite{kleinberg2000small}.


\begin{figure*}[htbp]
\centerline{\includegraphics[width=1\textwidth]{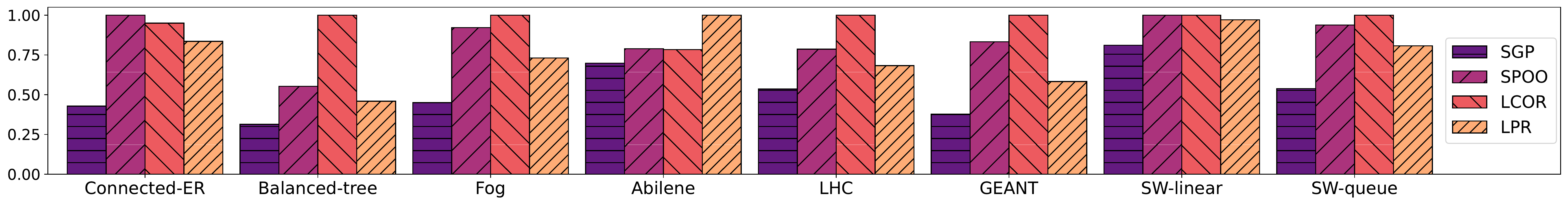}}
\vspace{-0.5\baselineskip}
\caption{Normalized total cost for network scenarios in Table \ref{tab_scenario}}
\label{fig_bar}
\vspace{-0.8\baselineskip}
\end{figure*}

\begin{figure*}
     \centering
     \begin{subfigure}[b]{0.215\textwidth}
         \centering
         \includegraphics[width=\textwidth]{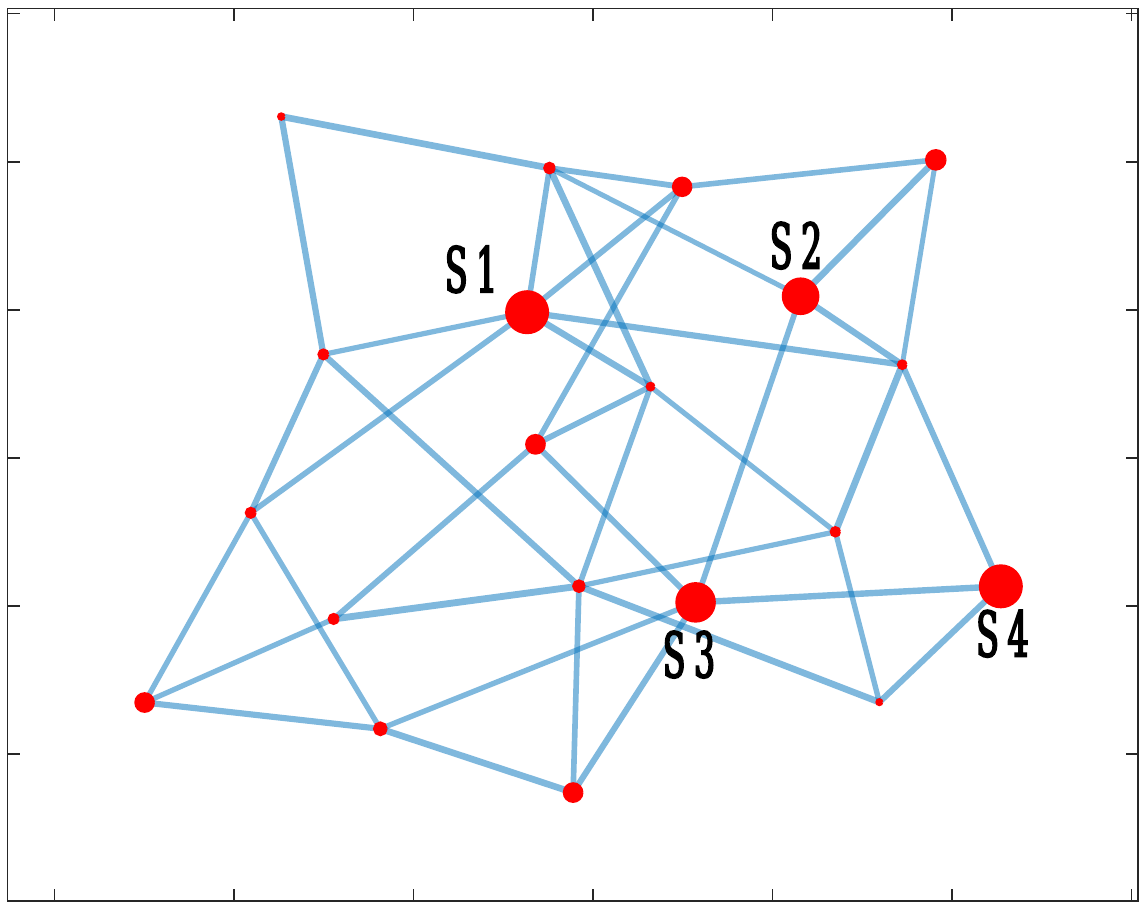}
         \caption{Topology \emph{Connected-ER}, link width equal to link capacity $d_{ij}$ and node size equal to computation capacity $s_i$}
         \vspace{-0.2\baselineskip}
         \label{fig:topology}
     \end{subfigure}
     \hfill
     \begin{subfigure}[b]{0.24\textwidth}
         \centering
         \includegraphics[width=\textwidth]{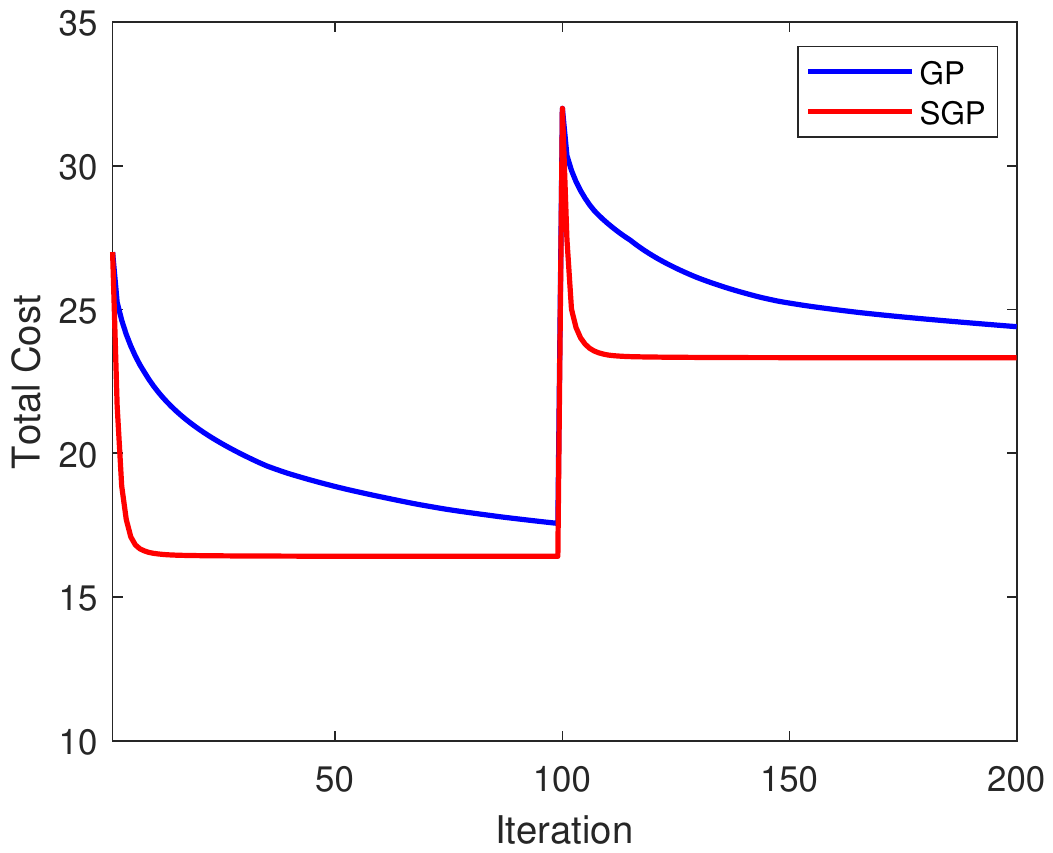}
         \vspace{-1.5\baselineskip}
         \caption{Convergence trajectory of \texttt{GP} and \texttt{SGP} subject to server failure at S1}
         \label{fig:convergence}
     \end{subfigure}
      \hfill
     \begin{subfigure}[b]{0.25\textwidth}
         \centering
         \includegraphics[width=\textwidth]{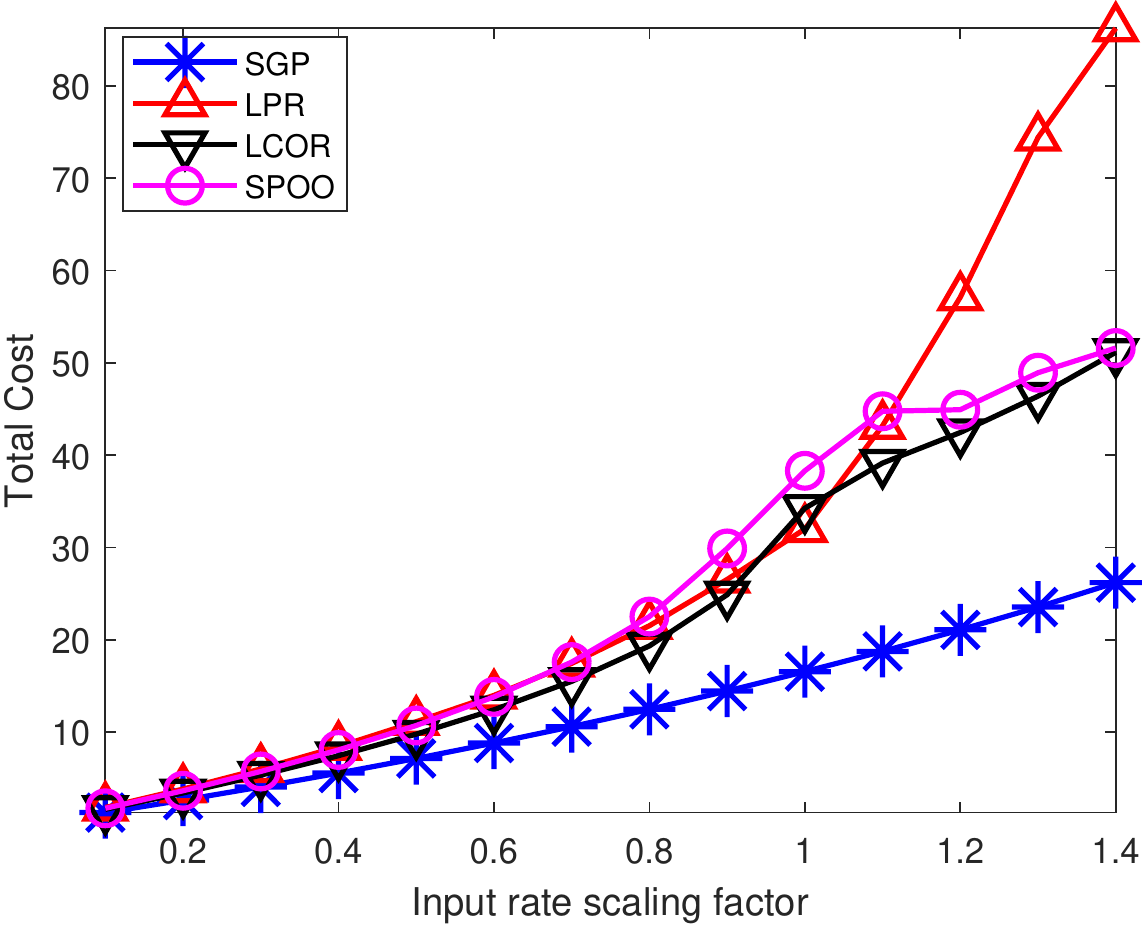}
         \caption{ Total cost versus scaled input rate}
         \label{fig:inputRate}
     \end{subfigure}
     \hfill
     \begin{subfigure}[b]{0.265\textwidth}
         \centering
         \includegraphics[width=\textwidth]{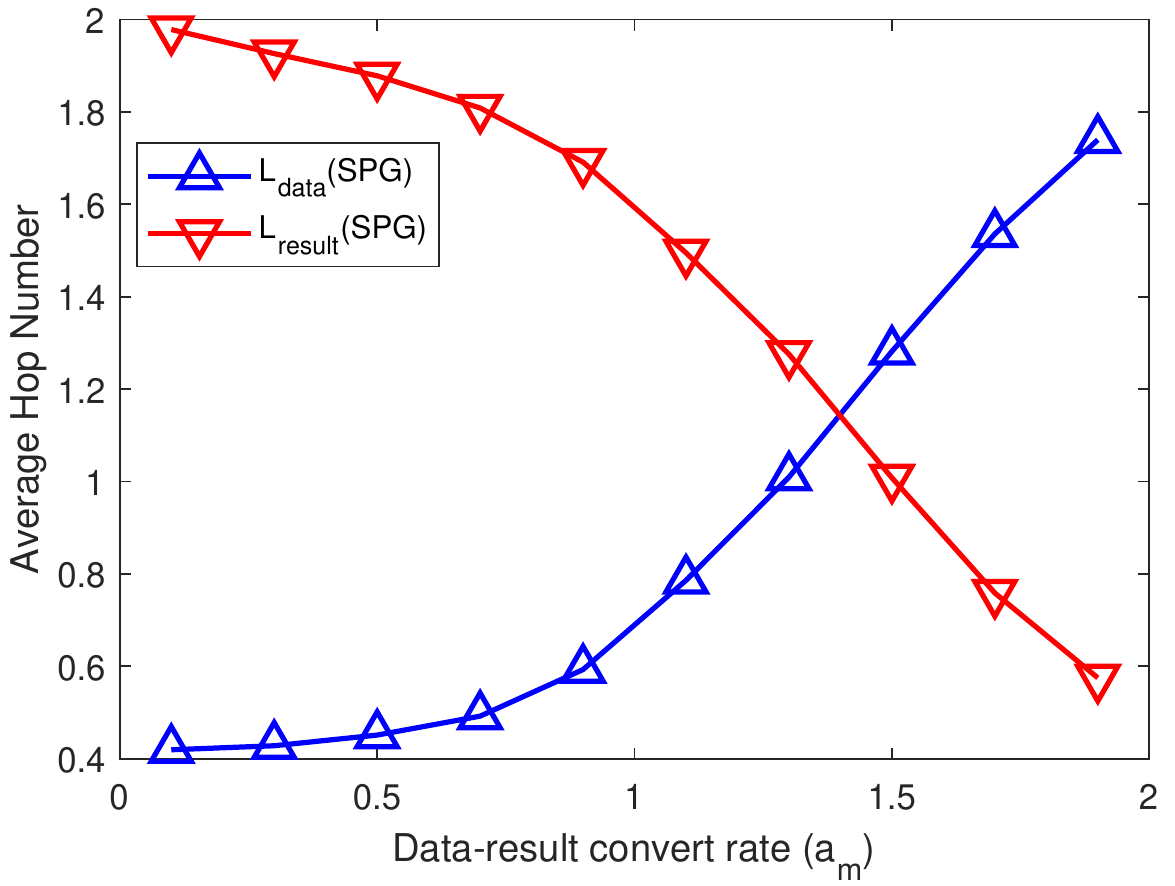}
         \caption{ $L_{\text{data}}$, $L_{\text{result}}$ and their ratio versus $a_m$}
         \label{fig:comploc}
     \end{subfigure}
      \vspace{-0.3\baselineskip}
        \caption{Topology and results in scenario \emph{Connected-ER}}
        
 \vspace{-1\baselineskip}
\end{figure*}

Table \ref{tab_scenario} also summarizes the number of nodes $|V|$ and edges $|E|$, as well as the number of 
tasks $|S|$ in each network.
We set $a_m$ to be exponential with mean value $0.5$ and truncated into interval $[0.1,5]$, considering that most computations have $a_m$ smaller that $1$, but special types like video rendering have relatively larger $a_m$.
Each task is randomly assigned with one computation type and one destination node, along with $|\mathcal{R}|$ random active data source (i.e. $r_{i}(d,m) > 0$). The input rate $r_{i}(d,m)$ of each active data source is 
chosen u.a.r.
in $[r_{\min},r_{\max}]$.
\textbf{Link} is the type of link cost $D_{ij}(\cdot)$, where \emph{Linear} denotes a linear link cost with unit cost $d_{ij}$, i.e. $D_{ij}(F_{ij}) = d_{ij}F_{ij}$, and \emph{Queue} denotes a queueing delay with link capacity $d_{ij}$, i.e. $D_{ij}(F_{ij}) = \frac{F_{ij}}{d_{ij} - F_{ij}}$.
\textbf{Comp} is the type of computation cost $C_i(\bf{G}_i)$, where \emph{Sum-Linear} denotes a weighted sum of linear cost for each type, i.e. $C_i(\boldsymbol{G_i}) = s_i\sum_{m}c_{im}g_{i}^m$, and \emph{Sum-Queue} denotes a queueing delay-like computation cost with capacity $s_i$, i.e. $C_i(\boldsymbol{G}_i) = \frac{\sum_{m} c_{im}g_{i}^m}{s_i - \sum_{m} c_{im}g_{i}^m}$, where the weights $c_{im}$ is 
u.a.r.
drawn from $[1,5]$.
The parameters $d_{ij}$ are 
u.a.r.
drawn from $[0, 2\Bar{d}_{ij}]$ and truncated into $[d_{\min},d_{\max}]$. Parameter $s_i$ are exponential random variables with mean $\Bar{s}_i$ truncated into $[s_{\min},s_{\max}]$ for \emph{Sum-Queue}, or uniform with mean $\Bar{s}_i$ for \emph{Sum-Linear}.

We implement the following baseline algorithms. 
Since 
this paper is the first to study joint routing and computation partial offloading in arbitrary network topologies with congestion-dependent cost and non-negligible result size, we make adaptation to these baselines to fit our model.

\noindent\textbf{GP}(Gradient Projection): similar to \texttt{SGP} but with the scaling matrices $M^-_{i}(d,m)$ and $M^+_{i}(d,m)$ being identity matrix multiplied by universal stepsize $0.1$. \texttt{GP} and \texttt{SGP} 
lead to the same global strategy but with different convergence speed.

\noindent\textbf{SPOO}(Shortest Path Optimal Offloading): 
fixes the routing variables $\phi^-_{ij}$,$\phi^+_{ij}$ with $j \neq 0$ 
to the shortest path (measured with marginal cost at $F_{ij} = 0$, accounting for the propagation delay without queueing effect), and studies the optimal offloading along these paths. Similar strategy is considered in \cite{he2021multi} with linear-topology and partial offloading. {Note that when destination and data source are the same node, \texttt{SPOO} is restricted to local computation.}

\noindent\textbf{LCOR}(Local Computation Optimal Routing): computes at the data sources (or with minimum offloading if pure local computation is not feasible), optimally route the result to destinations using scaled gradient projection in \cite{bertsekas1984second}.

\noindent\textbf{LPR}(Linear Program Rounded): the joint path-routing and offloading method by \cite{liu2020distributed}, which does not consider partial offloading, congestible links and result flow. 
To adapt \texttt{LPR}'s linear link costs to our schemes, we use the marginal cost at zero flow. 
To ensure sufficient communication resources for the result flow, we assign a saturate-factor of $0.7$ for queueing delay costs, i.e., the data flow could not exceed $0.7$ times real capacity.
Shortest path routing is used for result flow.



Fig.\ref{fig_bar} compares the total cost $D$ of different algorithms at the steady state over networks in Table \ref{tab_scenario} (we omit \texttt{GP} as it has the same steady state performance with \texttt{SGP}), where the bar heights of each scenario are normalized according to the worst algorithm. We test both linear cost and queueing delay with other parameters fixed in topology \texttt{SW}, labeled as \texttt{SW-linear} and \texttt{SW-queue}. 
Our proposed algorithm \texttt{SGP} significantly outperforms other baselines in all simulated scenarios, with more than $30\%$ improvement on average over \texttt{LPR}, which also jointly optimizes routing and task offloading but does not consider partial offloading and congestible links. 
The difference of case \texttt{SW-linear} and \texttt{SW-queue} suggests that our proposed algorithm promises a considerable improvement to SOTA especially when the networks are congestible.
Note that \texttt{LCOR} and \texttt{SPOO} reflects the optimal objective for routing and offloading subproblems, respectively. The gain of jointly optimizing over both strategies could be inferred by comparing \texttt{SGP} against \texttt{LCOR} and \texttt{SPOO}. For example, \texttt{LCOR} performs extremely bad in topology \emph{Balanced-tree}, because no routing could be optimized in a tree topology. 


We also perform refined experiments in \texttt{Connected-ER}, with the network topology and capacity shown in Fig.\ref{fig:topology}. There are $4$ major servers as labeled, and we assume server S1 fails (communication and computation capability disabled, stop sending data or making requests) at the $100$-th iteration.
We compare the convergence speed of \texttt{GP} and \texttt{SGP} in Fig.\ref{fig:convergence} subject to such server failure. 
\texttt{SGP} takes much less iterations to converge and adapt to topology change, showing the advantages of the sophisticatedly designed scaling matrices.

Fig.\ref{fig:inputRate} shows the change of total cost subject to universally scaled input rates $r_{i}(d,m)$, with other parameters fixed. The performance advantage of \texttt{SGP} has a rapid growth as the network getting more congested, especially against \texttt{LPR}.

To further illustrate why \texttt{SGP} outperforms baselines significantly with congestion-dependent cost, 
we define $L_{\text{data}}$ and $L_{\text{result}}$ as the average travel distance (hop number) of data packages from input to computation, and that of result packages from generation to being delivered, respectively.

In Fig. \ref{fig:comploc}, we compare $L_{\text{data}}$, $L_{\text{result}}$ for \texttt{SGP} 
over different universal $a_m$ with other parameters fixed.
The trajectories 
suggest that the average computation offloading distance grows with $a_m$. 
Namely, for tasks generating more result with unit input data, 
\texttt{SGP} tends to compute them nearer to the destination. 
Considering that
when $a_m \gg 1$, the network is highly congested mainly due to the result flow, thus the optimal strategy is to offload computation closer to destinations in order to reduce the result transmission distance and mitigate such congestion, in which case $L_{\text{data}}$ is large and $L_{\text{result}}$ is small.
Note that the trajectories also imply that the speed of growth of $L_{\text{data}}$ and descent of $L_{\text{result}}$ are low when $a_m$ is small. This is because when little result flow is generated, the network is lightly congested and the computation cost dominates the total cost, our algorithm tends to offload large portion of computation to the servers. However when $a_m$ is sufficiently large, the transmission cost overwhelms computation cost, then \texttt{SGP} speeds up shifting computation sites closer to destination.
The above behavior demonstrates the underlying optimality of our proposed method, namely reaching a ``balance'' among the cost for data forwarding, result forwarding and computation, and therefore optimizes the overall cost. 
As a comparison, \texttt{LPR} does not consider these aspects, the solution hardly changes with $a_m$, implying the cost of forwarding results could grow extremely high with large $a_m$.

\section{Conclusion} \label{Section:Conclusion}
We propose a novel joint routing and computation offloading model incorporating the result flow, partial offloading and multi-hop routing for both data and result. 
This is also the first flow model analysis of computation offloading adopting congestion-dependent link cost and arbitrary network topology.
We propose a total cost minimization problem to decide optimal routing-computation strategy. We optimally solve this non-convex problem by providing necessary and sufficient optimality conditions, and devise a fully distributed and scalable algorithm that reaches the global optimal.

\bibliographystyle{IEEEtran}
\bibliography{Reference}

\begin{thebibliography}{10}
\providecommand{\url}[1]{#1}
\csname url@samestyle\endcsname
\providecommand{\newblock}{\relax}
\providecommand{\bibinfo}[2]{#2}
\providecommand{\BIBentrySTDinterwordspacing}{\spaceskip=0pt\relax}
\providecommand{\BIBentryALTinterwordstretchfactor}{4}
\providecommand{\BIBentryALTinterwordspacing}{\spaceskip=\fontdimen2\font plus
\BIBentryALTinterwordstretchfactor\fontdimen3\font minus
  \fontdimen4\font\relax}
\providecommand{\BIBforeignlanguage}[2]{{%
\expandafter\ifx\csname l@#1\endcsname\relax
\typeout{** WARNING: IEEEtran.bst: No hyphenation pattern has been}%
\typeout{** loaded for the language `#1'. Using the pattern for}%
\typeout{** the default language instead.}%
\else
\language=\csname l@#1\endcsname
\fi
#2}}
\providecommand{\BIBdecl}{\relax}
\BIBdecl

\bibitem{Ericsson2021report}
\BIBentryALTinterwordspacing
Ericsson. Ericsson mobility report (2021, {Nov.}). [Online]. Available:
  \url{https://www.ericsson.com/en/reports-and-papers/mobility-report}
\BIBentrySTDinterwordspacing

\bibitem{sahni2020multi}
Y.~Sahni, J.~Cao, L.~Yang, and Y.~Ji, ``Multi-hop multi-task partial
  computation offloading in collaborative edge computing,'' \emph{IEEE
  Transactions on Parallel and Distributed Systems}, vol.~32, no.~5, pp.
  1133--1145, 2020.

\bibitem{sahni2017edge}
Y.~Sahni, J.~Cao, S.~Zhang, and L.~Yang, ``Edge mesh: A new paradigm to enable
  distributed intelligence in internet of things,'' \emph{IEEE access}, vol.~5,
  pp. 16\,441--16\,458, 2017.

\bibitem{zhu2017socially}
K.~Zhu, W.~Zhi, X.~Chen, and L.~Zhang, ``Socially motivated data caching in
  ultra-dense small cell networks,'' \emph{IEEE Network}, vol.~31, no.~4, pp.
  42--48, 2017.

\bibitem{hong2019multi}
Z.~Hong, W.~Chen, H.~Huang, S.~Guo, and Z.~Zheng, ``Multi-hop cooperative
  computation offloading for industrial iot--edge--cloud computing
  environments,'' \emph{IEEE Transactions on Parallel and Distributed Systems},
  vol.~30, no.~12, pp. 2759--2774, 2019.

\bibitem{liu2020distributed}
B.~Liu, Y.~Cao, Y.~Zhang, and T.~Jiang, ``A distributed framework for task
  offloading in edge computing networks of arbitrary topology,'' \emph{IEEE
  Transactions on Wireless Communications}, vol.~19, no.~4, pp. 2855--2867,
  2020.

\bibitem{al2016distributed}
H.~Al-Shatri, S.~M{\"u}ller, and A.~Klein, ``Distributed algorithm for energy
  efficient multi-hop computation offloading,'' in \emph{2016 IEEE
  International Conference on Communications (ICC)}.\hskip 1em plus 0.5em minus
  0.4em\relax IEEE, 2016, pp. 1--6.

\bibitem{sahni2018data}
Y.~Sahni, J.~Cao, and L.~Yang, ``Data-aware task allocation for achieving low
  latency in collaborative edge computing,'' \emph{IEEE Internet of Things
  Journal}, vol.~6, no.~2, pp. 3512--3524, 2018.

\bibitem{luo2021qoe}
Q.~Luo, W.~Shi, and P.~Fan, ``Qoe-driven computation offloading: Performance
  analysis and adaptive method,'' in \emph{2021 13th International Conference
  on Wireless Communications and Signal Processing (WCSP)}.\hskip 1em plus
  0.5em minus 0.4em\relax IEEE, 2021, pp. 1--5.

\bibitem{shi2019area}
W.~Shi, J.~Zhang, R.~Zhang, and K.~Hu, ``An area-based offloading policy for
  computing offloading in mec-assisted wireless mesh network,'' in \emph{2019
  IEEE/CIC International Conference on Communications in China (ICCC)}.\hskip
  1em plus 0.5em minus 0.4em\relax IEEE, 2019, pp. 507--511.

\bibitem{he2021multi}
X.~He, R.~Jin, and H.~Dai, ``Multi-hop task offloading with on-the-fly
  computation for multi-uav remote edge computing,'' \emph{IEEE Transactions on
  Communications}, 2021.

\bibitem{funai2019computational}
C.~Funai, C.~Tapparello, and W.~Heinzelman, ``Computational offloading for
  energy constrained devices in multi-hop cooperative networks,'' \emph{IEEE
  Transactions on Mobile Computing}, vol.~19, no.~1, pp. 60--73, 2019.

\bibitem{mcmahan2017communication}
B.~McMahan, E.~Moore, D.~Ramage, S.~Hampson, and B.~A. y~Arcas,
  ``Communication-efficient learning of deep networks from decentralized
  data,'' in \emph{Artificial intelligence and statistics}.\hskip 1em plus
  0.5em minus 0.4em\relax PMLR, 2017, pp. 1273--1282.

\bibitem{hong2019qos}
Z.~Hong, H.~Huang, S.~Guo, W.~Chen, and Z.~Zheng, ``Qos-aware cooperative
  computation offloading for robot swarms in cloud robotics,'' \emph{IEEE
  Transactions on Vehicular Technology}, vol.~68, no.~4, pp. 4027--4041, 2019.

\bibitem{liu2019joint}
B.~Liu, K.~Poularakis, L.~Tassiulas, and T.~Jiang, ``Joint caching and routing
  in congestible networks of arbitrary topology,'' \emph{IEEE Internet of
  Things Journal}, vol.~6, no.~6, pp. 10\,105--10\,118, 2019.

\bibitem{ren2019collaborative}
J.~Ren, G.~Yu, Y.~He, and G.~Y. Li, ``Collaborative cloud and edge computing
  for latency minimization,'' \emph{IEEE Transactions on Vehicular Technology},
  vol.~68, no.~5, pp. 5031--5044, 2019.

\bibitem{bertsekas2021data}
D.~Bertsekas and R.~Gallager, \emph{Data networks}.\hskip 1em plus 0.5em minus
  0.4em\relax Athena Scientific, 2021.

\bibitem{zhang2018optimal}
J.~Zhang, A.~Sinha, J.~Llorca, A.~Tulino, and E.~Modiano, ``Optimal control of
  distributed computing networks with mixed-cast traffic flows,'' in \emph{IEEE
  INFOCOM 2018-IEEE Conference on Computer Communications}.\hskip 1em plus
  0.5em minus 0.4em\relax IEEE, 2018, pp. 1880--1888.

\bibitem{chen2021edgeconomics}
Z.~Chen, Q.~Ma, L.~Gao, and X.~Chen, ``Edgeconomics: Price competition and
  selfish computation offloading in multi-server edge computing networks,'' in
  \emph{2021 19th International Symposium on Modeling and Optimization in
  Mobile, Ad hoc, and Wireless Networks (WiOpt)}.\hskip 1em plus 0.5em minus
  0.4em\relax IEEE, 2021, pp. 1--8.

\bibitem{gallager1977minimum}
R.~Gallager, ``A minimum delay routing algorithm using distributed
  computation,'' \emph{IEEE transactions on communications}, vol.~25, no.~1,
  pp. 73--85, 1977.

\bibitem{xi2008node}
Y.~Xi and E.~M. Yeh, ``Node-based optimal power control, routing, and
  congestion control in wireless networks,'' \emph{IEEE Transactions on
  Information Theory}, vol.~54, no.~9, pp. 4081--4106, 2008.

\bibitem{kamran2019deco}
K.~Kamran, E.~Yeh, and Q.~Ma, ``Deco: Joint computation, caching and forwarding
  in data-centric computing networks,'' in \emph{Proceedings of the Twentieth
  ACM International Symposium on Mobile Ad Hoc Networking and Computing}, 2019,
  pp. 111--120.

\bibitem{rossi2011caching}
D.~Rossi and G.~Rossini, ``Caching performance of content centric networks
  under multi-path routing (and more),'' \emph{Relat{\'o}rio t{\'e}cnico,
  Telecom ParisTech}, vol. 2011, pp. 1--6, 2011.

\bibitem{kleinberg2000small}
J.~Kleinberg, ``The small-world phenomenon: An algorithmic perspective,'' in
  \emph{Proceedings of the thirty-second annual ACM symposium on Theory of
  computing}, 2000, pp. 163--170.

\bibitem{bertsekas1984second}
D.~Bertsekas, E.~Gafni, and R.~Gallager, ``Second derivative algorithms for
  minimum delay distributed routing in networks,'' \emph{IEEE Transactions on
  Communications}, vol.~32, no.~8, pp. 911--919, 1984.

\end{thebibliography}

\section*{Appendix}
\subsection*{Proof of Theorem \ref{Thm_Sufficient}}
For simplicity, we consider the non-destination nodes in this proof, namely we assume $\sum_{j} \phi_{ij}^+ = 1$ for all $i$, while the derivation is applicable to destination nodes.
We have 
\begin{align*}
    \frac{\partial D}{\partial r_i(d,m)} &= \sum_{j \in \left\{0\right\} \cup \mathcal{O}(i)} \phi_{ij}^-(d,m) \delta_{ij}^-(d,m) 
    \\ &= \sum_{j : \phi_{ij} > 0} \phi_{ij}^-(d,m) \lambda_{idm}^- 
    \\ &= \lambda_{idm}^-,
\end{align*}
and thus
\begin{align}
    \delta_{ij}^-(d,m) \geq \frac{\partial D}{\partial r_i(d,m)}, \forall i \in V, j \in \left\{0 \right\} \cup \mathcal{O}(i), \forall (d,m) \in S. \label{SuffCond_alt1}
\end{align}

Similarly we have
\begin{align}
    \delta_{ij}^+(d,m) \geq \frac{\partial D}{\partial t_i^+(d,m)}, \quad \forall i \in V, j \in \mathcal{O}(i), \forall (d,m) \in S. \label{SuffCond_alt2}
\end{align}

To prove $ \phi = (\boldsymbol{\phi}^-,\boldsymbol{\phi}^+)$ minimizes $D = \sum_{(i,j) \in E}D_{ij}(F_{ij}) + \sum_{i \in V}C_i(\boldsymbol{G_i})$, let $\phi^* = (\phi^{-*}, \phi^{+*}) \neq \phi$ be another set of variable, and with corresponding forwarding and computation flows $F_{ij}^*, \forall (i,j) \in E$ and $\boldsymbol{G_i}^*, \forall i \in V$.
Given both $\phi$ and $\phi^*$ are valid forwarding and computing scheme, we know $(F_{ij},\boldsymbol{G_i})$ and $(F_{ij}^*,\boldsymbol{G_i}^*)$ are in the feasible set of the flow model problem (\ref{JointProblem}), which is a convex polytope. 

\begin{align}
    \min_{f^-,f^+,g} \quad &\sum_{(i,j) \in E} D_{ij}(F_{ij}) + \sum_{i \in V}C_i(\boldsymbol{G_i}) \label{JointProblem}
    \\ \text{such that} \quad &  \text{(\ref{FlowConservation1.1}) to (\ref{FlowConservation2.2}) hold,} \nonumber
    \\ & g_i(d,m) \geq 0, \quad \forall i \in V, k \leq M, (d,m) \in S \nonumber
    \\ & \begin{cases}f_{ij}^-(d,m) \geq 0,
    \\ f_{ij}^+(d,m) \geq 0, 
    \end{cases}  \, \forall (i,j) \in E, (d,m) \in S \nonumber
\end{align}

Due to the convexity of the feasible set, for any $\mu \in [0,1]$, $\left( (1 - \mu) F_{ij} + \mu F_{ij}^* ,  (1 - \mu) \boldsymbol{G_i} + \mu \boldsymbol{G_i}^* \right)$ is also feasible for (\ref{JointProblem}), we then let
\begin{align*}
   D(\mu) &= \sum_{(i,j) \in E}D_{ij}(  (1 - \mu) F_{ij} + \mu F_{ij}^* ) 
   \\ &+ \sum_{i \in V}C_i((1 - \mu) \boldsymbol{G_i} + \mu \boldsymbol{G_i}^*). 
\end{align*}
Since $D$ is convex in $F_{ij}$ and $\boldsymbol{G}_{i}$, we know $D(\mu)$ is convex in $\mu$. Thus combining with the arbitrary choice of $\phi^*$, the sufficiency in Theorem \ref{Thm_Sufficient} is proved if $\frac{d D(\mu)}{d \mu}$ is non-negative at $\mu = 0$. That is, we will show the following is non-negative
\begin{equation}
\begin{aligned}
    \frac{d D(\mu)}{ d \mu} \bigg|_{\mu = 0} &= \sum_{(i,j) \in E}D^\prime_{ij}(F_{ij}) (F_{ij}^* - F_{ij}) 
    \\ &+ \sum_{i \in V} \sum_{m \leq M} \frac{\partial C_i (\boldsymbol{G_i})}{\partial g_i^m} (g_i^{m*} - g_i^{m}) 
\end{aligned}    
\label{proof_obj}
\end{equation}

Starting with the data flow, multiply both side of (\ref{SuffCond_alt1}) by $\phi_{ij}^{-*}(d,m)$ and sum over $j \in \left\{ 0 \right\} \cup \mathcal{O}(i) $, we have
\begin{equation}
\begin{aligned}
    &\frac{\partial C_i(\boldsymbol{G_i})}{\partial g_i^m} \phi_{i0}^{-*}(d,m) + \sum_{j \in \mathcal{O}(i)} D_{ij}^\prime(F_{ij}) \phi_{ij}^{-*}(d,m) 
    \\ \geq &\frac{\partial D}{\partial r_i(d,m)} - a_m \frac{\partial D}{\partial t_i^+(d,m)} \phi_{i0}^{-*}(d,m) 
    \\ &- \sum_{j \in \mathcal{O}(i)} \frac{ \partial D}{\partial r_j(d,m)} \phi_{ij}^{-*}(d,m),
\end{aligned}    
\label{proof_ineqToEq1}
\end{equation}
then multiply both side by $t_i^{-*}(d,m) = \sum_{j \in \mathcal{I}(i)}f_{ji}^{-*}(d,m) + r_i(d,m)$, we have
\begin{align*}
     &\frac{\partial C_i(\boldsymbol{G_i})}{\partial g_i^m} g_i^{*}(d,m) + \sum_{j \in \mathcal{O}(i)} D_{ij}^\prime(F_{ij}) f_{ij}^{-*}(d,m)
    \\ \geq &  t_i^{-*}(d,m) \frac{\partial D}{\partial r_i(d,m)} - a_m \frac{\partial D}{\partial t_i  ^+(d,m)} t_i^{-*}(d,m) \phi_{i0}^{-*}(d,m) 
    \\ - &\sum_{j \in \mathcal{O}(i)} \frac{ \partial D}{\partial r_j(d,m)} t_i^{-*}(d,m) \phi_{ij}^{-*}(d,m),
\end{align*}
further sum over $(d,m) \in S$ and $j \in V$, we get
\begin{equation}
\begin{aligned}
    &\sum_{i \in V} \sum_{m \leq M} \frac{\partial C_i(\boldsymbol{G_i})}{\partial g_i^m} g_i^{m*} + \sum_{(i,j) \in E} D_{ij}^\prime(F_{ij}) F_{ij}^{-*}
    \\ &\geq \sum_{i \in V} \sum_{(d,m) \in S} t_i^{-*}(d,m) \frac{\partial D}{\partial r_i(d,m)} 
    \\ &- \sum_{i \in V} \sum_{(d,m) \in S}  a_m  \frac{\partial D}{\partial t_i^+(d,m)} t_i^{-*}(d,m) \phi_{i0}^{-*}(d,m)
    \\  &- \sum_{(d,m) \in S} \sum_{i \in V} \sum_{j \in \mathcal{O}(i)} \frac{ \partial D}{\partial r_j(d,m)} t_i^{-*}(d,m) \phi_{ij}^{-*}(d,m),
\end{aligned}   
\label{proof_sum_1}
\end{equation}
where $F_{ij}^{-*} = \sum_{(d,m) \in S} f_{ij}^{-*}(d,m)$.

Meanwhile, by the flow conservation (\ref{FlowConservation1.1}) to (\ref{FlowConservation2.2}), we know that for all $j \in V, (d,m) \in S$,
\begin{align*}
    \sum_{i \in \mathcal{I}(j)} t_i^{-*}(d,m) \phi_{ij}^{-*}(d,m) = t_j^{-*}(d,m) - r_j(d,m).
\end{align*}
Substitute above into the very last term in (\ref{proof_sum_1}) and cancel, we get
\begin{equation}
\begin{aligned}
    &\sum_{i \in V} \sum_{m \leq M} \frac{\partial C_i(\boldsymbol{G_i})}{\partial g_i^m} g_i^{m*} + \sum_{(i,j) \in E} D_{ij}^\prime(F_{ij}) F_{ij}^{-*} 
    \\  & \geq \sum_{i \in V} \sum_{(d,m) \in S} r_i(d,m)\frac{\partial D}{\partial r_i(d,m) } 
    \\ &- \sum_{i \in V} \sum_{(d,m) \in S} a_m g_i^{m*} \frac{\partial D}{\partial t_i^+(d,m)}.
\end{aligned}
\label{proof_sum_1_final}
\end{equation}

Next, about the flow of computation result, multiply both side of (\ref{SuffCond_alt2}) by $\phi_{ij}^{+*}(d,m)$ and sum over $j \in \mathcal{O}(i) $, we have
\begin{equation}
\begin{aligned}
    &\sum_{j \in \mathcal{O}(i)}D_{ij}^\prime(F_{ij}) \phi_{ij}^{+*}(d,m) 
    \\ \geq & \frac{\partial D}{\partial t_i^+(d,m)} - \sum_{j \in \mathcal{O}(i)} \frac{\partial D}{\partial t_j^+(d,m)} \phi_{ij}^{+*}(d,m).
\end{aligned}    
\label{proof_ineqToEq2}
\end{equation}

Multiply both side by $t_i^{+*}(d,m) = \sum_{j \in \mathcal{I}(i)}f_{ji}^{+*}(d,m) + a_m g_i^{*}(d,m)$, sum over $(d,m) \in S$ and $j \in V$, we get
\begin{equation}
\begin{aligned}
    &\sum_{(i,j) \in E}D_{ij}^\prime(F_{ij}) F_{ij}^{+*} 
    \\ &\geq \sum_{i \in V} \sum_{(d,m) \in S} t_i^{+*}(d,m) \frac{\partial D}{\partial t_i^+(d,m)} 
    \\ &- \sum_{(d,m)\in S} \sum_{i \in V} \sum_{j \in \mathcal{O}(i)} t_i^{+*}(d,m) \frac{\partial D}{\partial t_j^+(d,m)} \phi_{ij}^{+*}(d,m) .
\end{aligned}    
 \label{proof_sum_2}
\end{equation}

By (\ref{FlowConservation1.1}) to (\ref{FlowConservation2.2}), we have for all $ j \in V, (d,m) \in S$,
\begin{align*}
    \sum_{i \in \mathcal{I}(j)} t_i^{+*}(d,m) \phi_{ij}^{+*}(d,m) = t_j^{+*}(d,m) - a_m g_j^{*}(d,m). 
\end{align*}

Substituting above into the very last term in (\ref{proof_sum_2}) and canceling, we get
\begin{align}
    \sum_{(i,j) \in E} D_{ij}^\prime(F_{ij}) F_{ij}^{+*} \geq \sum_{i \in V} \sum_{(d,m) \in S} a_m g_i^{m*} \frac{\partial D}{\partial t_i^+(d,m)}.
     \label{proof_sum_2_final}
\end{align}

Summing up both side of (\ref{proof_sum_1_final}) and (\ref{proof_sum_2_final}), we have
\begin{equation}
\begin{aligned}
     &\sum_{i \in V} \sum_{m \leq M} \frac{\partial C_i(\boldsymbol{G_i})}{\partial g_i^m} g_i^{m*} + \sum_{(i,j) \in E} D_{ij}^\prime(F_{ij}) F_{ij}^{*} 
     \\ & \geq \sum_{i \in V} \sum_{(d,m) \in S} \frac{\partial D}{\partial r_i(d,m)} r_i(d,m).
\end{aligned}    
 \label{proof_ineq}
\end{equation}

Note that the equality would always hold in (\ref{proof_ineqToEq1}) and (\ref{proof_ineqToEq2}) if we substitute $\phi^*$ with $\phi$ in the above reasoning, as a consequence of (\ref{partial_D_r}) and (\ref{partial_D_t}).
Thus we have the following analogue of (\ref{proof_ineq}),
\begin{equation}
\begin{aligned}
    &\sum_{i \in V} \sum_{m \leq M} \frac{\partial C_i(\boldsymbol{G_i})}{\partial g_i^m} g_i^{m} + \sum_{(i,j) \in E} D_{ij}^\prime(F_{ij}) F_{ij}
    \\& = \sum_{i \in V} \sum_{(d,m) \in S} \frac{\partial D}{\partial r_i(d,m)} r_i(d,m).
\end{aligned}    
\label{proof_eq}
\end{equation}

Abstracting (\ref{proof_eq}) from (\ref{proof_ineq}), we show (\ref{proof_obj}) and complete the proof.

\end{document}